\documentclass[runningheads]{llncs}
\usepackage{xcolor}

\usepackage{bm}
\usepackage{stmaryrd}
\usepackage{mathtools}
\usepackage{tikz}
\usepackage{nicefrac}
\usepackage{amssymb}
\usetikzlibrary{petri,calc,backgrounds}
\usepackage{hyperref}
\usepackage[capitalize]{cleveref}
\usepackage{pgfplots}
\pgfplotsset{compat=1.15}
\usepackage{enumitem}
\usepackage{multirow}
\usepackage{marvosym}
\usepackage{booktabs}
\usetikzlibrary{hobby, decorations.pathreplacing}
\pgfdeclarelayer{back}
\pgfsetlayers{back,main}
\usepackage{macros}
\usepackage{xspace}

\usepackage{thm-restate}
\usepackage{todonotes}
\hypersetup{
    colorlinks,
    linkcolor={blue!70!black},
    citecolor={blue!70!black},
    urlcolor={red!70!black}
}

\begin{document}

\title{Verifying generalised and structural soundness of workflow nets via relaxations}
\titlerunning{Verifying generalised and structural soundness of workflow nets}

\author{Michael Blondin\inst{1}\orcidID{0000-0003-2914-2734} \and
  Filip Mazowiecki\inst{2}\orcidID{0000-0002-4535-6508} \and
  Philip Offtermatt\inst{1,2}\orcidID{0000-0001-8477-2849}}

\institute{Universit\'{e} de Sherbrooke, Sherbrooke, Canada \and Max Planck Institute for Software Systems, Saarbr\"ucken, Germany}

\maketitle

\begin{abstract}
  Workflow nets are a well-established mathematical formalism for the
  analysis of business processes arising from either modeling tools or
  process mining. The central decision problems for workflow nets are
  $k$-soundness, generalised soundness and structural soundness. Most
  existing tools focus on $k$-soundness. In this work, we propose
  novel scalable semi-procedures for generalised and structural
  soundness. This is achieved via integral and continuous Petri net
  reachability relaxations. We show that our approach is competitive
  against state-of-the-art tools.
\end{abstract}

\section{Introduction}
\emph{Workflow nets} are a well-established mathematical formalism for the
description of business processes arising from software modelers and
process mining (\eg, see~\cite{van1998application,Aalst00}), and
further notations such as UML activity diagrams~\cite{Aalst2002}. More
precisely, a workflow net consists of
\emph{places} that contain resources, and \emph{transitions} that can
consume, create and move resources concurrently. Two designated
places, denoted $\initial$ and $\output$, respectively model the
initialization and completion of a process. Workflow nets, which form
a subclass of Petri nets, enable the automatic formal verification of
business processes. For example, \emph{$1$-soundness} states that from
the initial configuration $\imarked{1}$, every reachable configuration
can reach the final configuration $\fmarked{1}$. Informally, this
means that given any partial execution of a business process, it is
possible to complete it properly.


\paragraph*{Soundness.}

The main decision problems concerning workflow nets revolve around
soundness properties. The generalisation of $1$-soundness to several
resources is \emph{$k$-soundness}. It asks whether from $\imarked{k}$,
every reachable configuration can reach $\fmarked{k}$ (here,
$\{p \colon k\}$ indicates that place $p$ contains $k$
resources). \emph{Generalised soundness} asks whether $k$-soundness
holds for all $k \geq 1$. Unlike $k$-soundness, generalised soundness
preserves desirable properties like
composition~\cite{van2003soundness}. \emph{Structural soundness} is
the existential counterpart of generalised soundness, \ie\ it asks
whether $k$-soundness holds for some $k \geq 1$. These problems are
all decidable~\cite{AL97,HSV04,ctiplea2005structural}, but with high
complexity: either PSPACE- or EXPSPACE-complete~\cite{BMO22}. Most of
the (software) tools focus on $k$-soundness, with an emphasis on $k =
1$. Existing algorithms for generalised and structural soundness rely
on Petri net
reachability~\cite{HSV04,ctiplea2005structural,van2007verifying},
which was recently shown Ackermann-complete~\cite{LS19,Ler21,CO21}, so
not primitive recursive. In this work, we describe
\emph{novel scalable semi-procedures for generalised and
  structural soundness}.

We focus on ``negative instances'', \ie\ where soundness does
\emph{not} hold. Let us motivate this. It is known that given a
workflow net $\pn$, one can iteratively apply simple reduction rules
to $\pn$. The resulting workflow net $\pn'$ is sound iff
$\pn$ is as well~\cite{bride2017reduction,hoffmann2017workflow}. In
practice, one infers that $\pn$ is sound from
the fact that $\pn'$ has been reduced to a trivial workflow net where
only $\initial$ and $\output$ remain. However, if $\pn$ is
\emph{not} sound, one obtains some nontrivial $\pn'$ that
must be verified via some other approach such as model checking. In
this work, we provide algorithmic building blocks for this case, where state-space exploration is prohibitive.

\paragraph*{Relaxations.}

This is achieved by considering two reachability relaxations, namely
integer reachability and continuous reachability. As their name
suggests, these two notions relax some forbidden behaviour of workflow
nets. Informally, integer reachability allows for the amount of
resources to become temporarily negative, while continuous
reachability allows the fragmentation of resources into pieces. Such
relaxations possibly introduce spurious behaviour, but enjoy
significantly better algorithmic properties (\eg,
see~\cite{Blo20}). For example, they have been successfully employed
for the verification of multi-threaded program
skeletons~\cite{ELMMN14,ALW16,BFHH17}.

\paragraph*{Generalised soundness.}

Based on these
relaxations, we provide two necessary conditions for generalised
soundness: \emph{integer boundedness} and \emph{continuous
  soundness}. The former states that the state-space of a given
workflow net is bounded (from above) even under integer
reachability. The latter states that a given workflow net is $1$-sound
under continuous reachability. We show the following for integer
boundedness and continuous soundness:
\begin{itemize}
\item Well-established classical reduction rules preserve both
  properties;
  
\item Integer boundedness is testable in polynomial time,
and continuous soundness is coNP-complete;

\item From a practical viewpoint, they are respectively translatable
  into instances of linear programming and linear arithmetic (which
  can be solved efficiently by dedicated tools such as SMT solvers);

\item Under a mild computational assumption, continuous soundness
  implies integer boundedness.
\end{itemize}
  
Thus, altogether, in order to check whether a workflow net $\pn$ is
generalised \emph{unsound}, one may first use classical reduction
rules to obtain a smaller workflow net $\pn'$; test integer
\emph{unboundedness} in polynomial time; and, if needed, move onto
testing continuous \emph{unsoundness}.

The fact that continuous reachability can be used to semi-decide
generalised soundness is arguably surprising. Using the notation of
computation temporal logic (CTL), $k$-soundness can be rephrased as
$\imarked{k} \models \forall \mathsf{G}\, \exists \mathsf{F}\,
\fmarked{k}$. Some other well-studied properties have a similar structure,
\eg\ liveness and home-stateness amount to ``$\m_\text{init} \models
\bigwedge_{t \in T} \forall \mathsf{G}\, \exists \mathsf{F}\,
(t~\text{is enabled})$'' and ``$\m_\text{init} \models \forall
\mathsf{G}\, \exists \mathsf{F}\, \m_\text{home}$''. It is known that
liveness, home-stateness, and other properties such as boundedness and
inclusion, \emph{cannot} be approximated
continuously~\cite[Sect.~4]{BH17}. Yet, generalised soundness
quantifies $k$-soundness universally, and this enables a continuous
over-approximation. Consequently, we provide a novel application of
continuous relaxations for the efficient verification of properties
beyond reachability.

\paragraph*{Structural soundness.}

The authors of~\cite{ctiplea2005structural} have observed that a
property called structural quasi-soundness is a necessary condition for
structural soundness. The former states that $\imarked{k}$ can reach
$\fmarked{k}$ for some $k \geq 1$. In~\cite{ctiplea2005structural},
structural quasi-soundness is reduced to Petri net reachability, which
has non primitive recursive complexity. In this work, we show that
structural quasi-soundness can be rephrased as continuous
reachability. Since the latter can be tested in polynomial
time~\cite{FH15}, or alternatively via SMT solving~\cite{BFHH17}, this
vastly improves the practicability of structural quasi-soundness. We
further show that this approach can be adapted so that it provides a
lower bound on the first $k$ such that $\imarked{k}$ can reach
$\imarked{f}$. From a practical point of view, this is useful as it
can vastly reduce the number of reachability queries to decide
structural soundness.

\paragraph*{Free-choice nets.}

Many real-world workflow nets have a specific structure where
concurrency is restricted. Such nets are known as \emph{free-choice}
workflow nets (\eg, see~\cite{DE95} for a book). In particular,
free-choice workflow nets allow for the modeling of many features
present in common workflow management
systems~\cite{van1998application}. Generalised soundness is equivalent
to $1$-soundness for free-choice workflow nets~\cite{ping2004on}. In
this work, we prove that continuous soundness is equivalent to
generalised soundness. As a byproduct of our proof, we show that
structural soundness is also equivalent to continuous
soundness. Altogether, the notions of $\{$\text{1-}$, \allowbreak
\text{generalised}, \allowbreak \text{structural}, \allowbreak
\text{continuous}\}$ soundness \emph{all coincide} for free-choice
nets. In particular, this means that the continuous relaxation is
\emph{exact} and can serve as an efficient addition to the existing
algorithmic toolkit.

\paragraph*{Experimental results.}

To demonstrate the viability of our approach, we have implemented and
experimentally evaluated a prototype. As part of our evaluation, we
propose several new synthetic instances for generalised and structural
soundness, which are hard to decide with naive approaches. Some of
these instances involve the composition of workflow nets arising from
the modeling of business processes in the IBM WebSphere Business
Modeler. Our prototype is competitive against both a state-of-the-art
Petri net model checker, and a workflow net analyzer. In particular,
our approach exhibits better signs of scalability.

\paragraph*{Organization.}

The paper follows the structure of this introduction.
\Cref{sec:prelims} introduces notation, workflow nets and some
properties. \Cref{sec:relaxations} defines integer and continuous
relaxations, and further shows that they are preserved under reduction
rules. \Cref{sec:gen:sound,sec:struct:sound,sec:freechoice} present
the aforementioned results on generalised soundness, structural
soundness and free-choice nets. \Cref{sec:experimental} provides
experimental results. \Cref{sec:conclusion} concludes. Some proofs are
deferred to an appendix.

\section{Preliminaries}
\label{sec:prelims}
We use $\Z$, $\N$, $\Q$ and $\Qpos$ to respectively denote the
integers, the naturals (including $0$), the rationals and the
nonnegative rationals (including $0$). Let $\vec{x}, \vec{y} \in \Q^S$
be vectors over a finite set $S$. We write $\vec{x} \leq \vec{y}$ if
$\vec{x}[s] \leq \vec{y}[s]$ for all $s \in S$. We write $\vec{x}
< \vec{y}$ if $\vec{x} \leq \vec{y}$ and $\vec{x}[s] < \vec{y}[s]$ for
some $s \in S$. We extend addition and subtraction to vectors, \ie\
$(\vec{x} + \vec{y})[s] \defeq \vec{x}[s] + \vec{y}[s]$ and $(\vec{x}
- \vec{y})[s] \defeq \vec{x}[s] - \vec{y}[s]$ for all $s \in S$.
We define $\support{\vec{x}} = \set{s \in S \mid \vec{x}[s] \neq 0}$.
Given $c \in \Q$,
$\vec{c} \in \Q^S$ denotes the vector such that $\vec{c}[s] = c$ for
all $s \in S$.

\subsection{Petri nets}

A \emph{Petri net} $\pn$ is a triple $(P, T, F)$, where $P$ is a
finite set of \emph{places}; $T$ is a finite set
of \emph{transitions}, such that $T \cap P = \emptyset$; and $F \colon
((P \times T) \cup (T \times P)) \to \set{0, 1}$ is a set
of \emph{arcs}. For readers familiar with Petri nets, note that arc weights are not allowed, \ie the weights are always
$1$. A \emph{marking} is a vector $\m \in 
\N^P$ such that $\m[p]$ denotes the number of \emph{tokens} in
place $p$. 
We denote markings listing nonzero values,
\eg $\m = \set{p_1 \colon 1}$ means $\m[p_1] = 1$ and $\m[p] = 0$ for $p \neq p_1$.

Let $t \in T$. We define the \emph{pre-vector} of $t$ as $\pre{t}
\in \N^P$, where $\pre{t}[p] \defeq F(p,t)$.
We define its \emph{post-vector}
symmetrically with $\post{t}[p] \defeq F(t, p)$. The
\emph{effect} of $t$ is denoted as $\effect{t} \defeq \post{t} -
\pre{t}$. We say that a transition $t$ is \emph{enabled} at a marking
$\m$ if $\m \geq \pre{t}$. If this is the case, then $t$ can be
\emph{fired} at $\m$, which results in a marking $\m'$ such that $\m'
\defeq \m + \effect{t}$. We write $\m \trans{t}$ to denote that $t$ is
\emph{enabled} at $\m$, and we write $\m \trans{t} \m'$ whenever we care
about the marking $\m'$ resulting from the firing. We further write
$\m \trans{} \m'$ to denote that $\m \trans{t} \m'$ for some $t \in
T$.

We say that a sequence of transitions $\pi = t_1 \cdots t_n$ is a
\emph{run}.  We extend the notion of effect, enabledness and firing
from transitions to runs in a straightforward way.  The \emph{effect}
of a run is defined as the sum of the effects of its transitions, that
is, $\effect{\pi} \defeq \effect{t_1} + \ldots + \effect{t_n}$. The
run $\pi$ is enabled at $\m$, denoted as $\m \trans{\pi}$, if $\m
\trans{t_1} \m_1 \trans{t_2} \m_2 \cdots \trans{t_{n-1}} \m_{n-1}
\trans{t_{n}}$ for some markings $\m_1, \m_2, \dots,
\m_{n-1}$. Furthermore, firing $\pi$ from $\m$ leads to $\m'$, denoted
as $\m \trans{\pi} \m'$, if $\m \trans{\pi}$ and $\m' = \m +
\effect{\pi}$. We denote the reflexive and transitive closure of
${\trans{}}$ by ${\reach}$.

A pair $(\pn, \m)$, where $\pn$ is a Petri net and $\m$ is a marking
of $\pn$, is called a \emph{marked Petri net}. We write $\Reach{\pn,
  \m} \defeq \{\m' \mid \m \reach \m'\}$ to denote the set of markings
reachable from $\m$ in $\pn$.

A marked Petri net $(\pn, \m)$ is \emph{bounded} if there exists $b
\in \N$ such that $\m' \in \Reach{\pn, \m}$ implies $\m'[p] \leq b$
for all $p \in P$. It is further \emph{safe} if $b = 1$. We say
\emph{unbounded} and \emph{unsafe} for ``not bounded'' and ``not
safe''.

Sometimes, we argue about transformations on Petri nets
which take as an input a Petri net $\pn$
and output a Petri net $\pn'$.
We say that such a transformation \emph{preserves} some property
if $\pn$ satisfies that property iff $\pn'$ satisfies it.

\begin{figure}
  \centering
  \begin{tikzpicture}[auto, node distance=1cm, transform shape, scale=0.8]
    \node[place, label=above:$\initial$, tokens=1] (i) {};
    \node[transition, right of=i, label=above:$s$] (s) {};
    
    \node[place, above right=0cm and 0.5cm of s, label=above:$p_1$] (p1) {};
    \node[place, below right=0cm and 0.5cm of s, label=below:$p_2$] (p2) {};
    \node[transition, right of=p1, label=above:$t_1$] (t1) {};
    \node[transition, right of=p2, label=below:$t_2$] (t2) {};

    \node[place, right of=t1, label=above:$q_1$] (q1) {};
    \node[place, right of=t2, label=below:$q_2$] (q2) {};

    \node[transition, below right=0cm and 0.5cm of q1,
          label=above:$u$] (u) {};
    \node[place, right of=u, label=above:$\output$] (f) {};

    \path[->]
    (i) edge node {} (s)
    (s) edge node {} (p1)
    (s) edge node {} (p2)

    (p1) edge node {} (t1)
    (p2) edge node {} (t2)

    (t1) edge node {} (q1)
    (t2) edge node {} (q2)

    (q1) edge node {} (u)
    (q2) edge node {} (u)
    (u)  edge node {} (f)
    ;

    \node[place, label=above:$\initial$, tokens=2, right=1cm of f] (ii) {};
    \node[transition, right of=ii] (t1) {};
    \node[place, right of=t1] (p1) {};

    \node[transition, above right=0cm and 0.5cm of p1] (t2) {};
    \node[transition, below right=0cm and 0.5cm of p1] (t3) {};

    \node[transition, above=0.25cm of t2] (t2r) {};
    \node[transition, below=0.25cm of t3] (t3r) {};

    \node[place, right of=t2] (p2) {};
    \node[place, right of=t3] (p3) {};

    \node[transition, below right=0cm and 0.5cm of p2,
          label=above:$t$] (t4) {};
    \node[place, label=above:$\output$, right of=t4] (ff) {};

    \node[place, below of=t4] (p4) {};
    \node[transition, right of=p4] (t5) {};

    \path[->]
    (ii) edge node {} (t1)
    (t1) edge node {} (p1)

    (p1) edge node {} (t2)
    (p1) edge node {} (t3)
    (t2) edge node {} (p2)
    (t3) edge node {} (p3)
    
    (p2)  edge[out=90,  in=0] node {} (t2r)
    (p3)  edge[out=-90, in=0] node {} (t3r)
    (t2r) edge[out=180, in=90]  node {} (p1)
    (t3r) edge[out=180, in=-90] node {} (p1)

    (p2) edge node {} (t4)
    (p3) edge node {} (t4)
    (t4) edge node {} (ff)

    (t4) edge node {} (p4)
    (p4) edge node {} (t5)
    (t5) edge node {} (ff)
    ;

    \coordinate (m) at ($(f)!0.5!(ii)$);
    \coordinate[yshift=35pt]  (m1) at (m);
    \coordinate[yshift=-35pt] (m2) at (m);

    \draw (m1) -- (m2);
  \end{tikzpicture}
  
  \caption{Example of two Petri nets: respectively $\pn_\text{left}$
    and $\pn_\text{right}$.}\label{fig:pn}
\end{figure}
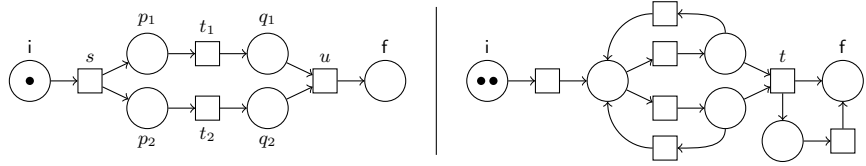

\begin{example}
  The left-hand side of \Cref{fig:pn} illustrates a Petri net
  $\pn_\text{left} = (P, T, F)$ where $P \defeq \{\initial, p_1, p_2,
  q_1, q_2, \output\}$, $T \defeq \{s, t_1, t_2, u\}$, and $F$ is
  depicted by arcs, \eg\ $F[\initial, s] = 1$ and $F[s, \initial] =
  0$. The Petri net is marked by $\imarked{1}$, \ie\ with one token in
  place $\initial$. We have $\imarked{1} \trans{s} \{p_1 \colon 1,
  p_2 \colon 1\} \trans{t_1 t_2} \{q_1 \colon 1, q_2 \colon
  1\} \trans{u} \fmarked{1}$.\hfill{$\lhd$}
\end{example}

\subsection{Workflow nets}

A workflow net $\pn$ is a Petri net~\cite{AL97} such that:
\begin{itemize}
\item there is a designated \emph{initial place} $\initial$ such that
  $\post{t}[\initial] = 0$ for all $t \in T$;

\item there is a designated \emph{final place} $\output \neq \initial$
  such that $\pre{t}[\output] = 0$ for all $t \in T$; and

\item each place and transition lies on at least one path from
  $\initial$ to $\output$ in the underlying graph of $\pn$, \ie\ $(V,
  E)$ where $V \defeq P \cup T$ and $(u, v) \in E$ iff $F(u, v) \neq 0$.
\end{itemize}
We say that $\pn$ is: 
\begin{itemize}
\item \emph{$k$-sound} if for all $\m \in \Reach{\pn, \imarked{k}}$ it
  is the case that $\m \reach \fmarked{k}$~\cite{AL97};

\item \emph{generalised sound} if $\pn$ is $k$-sound for all $k \in
  \Nn$~\cite[Def.~3]{van2003soundness},

\item \emph{structurally sound} if $\pn$ is $k$-sound for some $k \in \Nn$~\cite{barkaoui1998structural}.

\end{itemize}

\begin{example}
  \Cref{fig:pn} depicts two workflow nets: $\pn_\text{left}$ and
  $\pn_\text{right}$. The former is generalised sound, but the latter
  is not. Indeed, from $\imarked{1}$, transition $t$ cannot be enabled
  (as transitions preserve the sum of all tokens). Both workflow nets
  are structurally sound. Indeed, $\pn_\text{right}$ is $2$-sound as
  it is always possible to redistribute the two tokens so that $t$ can
  be fired in order to reach $\fmarked{2}$. \hfill{$\lhd$}
\end{example}

\section{Reachability relaxations}
\label{sec:relaxations}

Fix a Petri net $\pn = (P, T, F)$. We describe the two aforementioned
relaxations.

\paragraph{Integer reachability.}

An \emph{integral marking} is a vector $\m \in \Z^P$. Any
transition $t \in T$ is \emph{enabled} in $\m \in \Z^P$, and \emph{firing} $t$ leads
to $\m' \defeq \m + \effect{t}$, denoted $\m \ztrans{t} \m'$. We
define $\m \ztrans{} \m'$ and $\m \zreach \m'$ analogously to the
standard setting but w.r.t.\ $\ztrans{t}$ rather than
$\trans{t}$. Similarly, $\ZReach{\pn, \m} \defeq \{\m' \in \Z^P \mid
\m \zreach \m'\}$. As transitions are always enabled, the order of a
firing sequence is irrelevant. In particular, $\m \zreach \m'$ iff
there exists $\vec{x} \in \N^T$ such that $\m' = \m + \sum_{t \in T}
\vec{x}[t] \cdot \effect{t}$. Thus, integer reachability amounts to
integer linear programming. Moreover, it is
NP-complete~\cite{HH14,CHH18}.

\paragraph{Continuous reachability.}

A \emph{continuous marking} is a vector $\m \in \Qpos^P$. Let
$\lambda \in \ZeroOne$. We say that $\lambda t$ is \emph{enabled} in
$\m$, denoted $\m \ctrans{\lambda t}$, if $\m \geq \lambda
\cdot \pre{t}$. In this context, $\lambda$ is called the \emph{scaling
  factor}. Furthermore, we denote by $\m \ctrans{\lambda t} \m'$ that
$\lambda t$ is enabled in $\m$, and that its \emph{firing} results in
$\m' \defeq \m + \lambda \cdot \effect{t}$. A sequence of pairs of
scaling factors and transitions is called a \emph{continuous run}.

The notations $\m \ctrans{} \m'$ and $\m \creach \m'$ are defined
analogously to the discrete case but with respect to $\ctrans{\lambda
t}$ rather than $\trans{t}$ (the internal factors $\lambda$ can
differ). Similarly, $\QReach{\pn, \m}
\defeq \{\m' \mid \m \creach \m'\}$ denotes the markings continuously
reachable from $\m$. For example, for $\pn_\text{left}$
from \Cref{fig:pn} and $\pi \defeq \frac{1}{2} s\, \frac{1}{4} t_1$,
we have $\imarked{1} \ctrans{\pi} \{\initial \colon 1/2, p_1 \colon
1/4, p_2 \colon 1/2, q_1 \colon 1/4\}$. It is known that continuous
reachability, namely determining whether $\m \creach \m'$, given
$\m, \m' \in \Qpos^P$, can be checked in polynomial time~\cite{FH15}.




Let us establish the following helpful lemma similar
to~\cite[Lemma~12(1)]{FH15}.

\begin{restatable}{lemma}{lemContEquivReach}\label{lem:cont-equiv-reach}
  Let $\m$, $\m'$ be continuous markings. It is the case that $\m
  \creach \m'$ iff there exists $b \in \Nn$ such that $b \cdot \m
  \reach b \cdot \m'$.
\end{restatable}

\subsection{Preservation under reduction rules}
In~\cite{bride2017reduction}, the authors present six reduction rules,
denoted $R_1, \ldots, R_6$, that generalize the existing reduction
rules of~\cite{murata89petri}. In the following, we show that these reduction rules
preserve natural properties for the two reachability relaxations.
This means we will be able to check
these properties on a reduced workflow net and get the same results as on the original one.

Formally, the rules simplify a given workflow
net $\pn = (P, T, F)$. In particular, the places of the resulting
workflow net $\pn' = (P', T, F')$ form a subset of $P$. Let us fix a
domain $\D \in \{\N, \Z, \Qpos\}$ and let $P' \subseteq P$. For ease
of notation, we we write $P'' = P \setminus P'$ to denote the (possibly
empty) set of removed places. Rules never remove the initial and
output places, \ie\ $\initial, \output \in P'$. We denote by $\pi
\colon \D^{P} \to \D^{P'}$ the obvious projection function, and by $\pi_0
\colon \D^{P'} \to \D^{P}$ the ``reverse projection'' which fills new
places with $0$. Formally, $\pi_0(\m)[p'] \defeq \m[p']$ for all
$p' \in P'$ and $\pi_0(\m)[p''] \defeq 0$ for all $p'' \in P''$.



In~\cite{bride2017reduction}, the authors prove that the
rules preserve generalised soundness. This of course implies that they
preserve $k$-soundness for all $k$. The technical proposition below will
be helpful in the forthcoming sections
to show the preservation of useful properties based on reachability relaxations.

\begin{restatable}{proposition}{propReductions}\label{prop:reductions}
  Let $\pn = (P, T, F)$ be a workflow net, and let
  $\D \in \{\N, \Z, \Qpos\}$. Let $\pn' = (P', T', F')$ be a workflow
  net obtained by applying a reduction rule $R_i$ to $\pn$, where $P =
  P' \cup P''$. The following holds.

  \begin{itemize}

  \item \emph{Rule $R_1$.}\ We have $P'' = \set{p}$. There exists a
    nonempty set $R' \subseteq P'$ such that if $\imarked{1} \dreach
    \m$ in $\pn$, then $\m[p] = \sum_{r \in R'} \m[r']$. Moreover, $\m
    \dreach \n$ in $\pn$ iff $\pi(\m) \dreach \pi(\n)$ in $\pn'$.

  \item \emph{Rules $R_2$ and $R_3$.}\ We have $P'' = \emptyset$ and
    $\m \dreach \n$ in $\pn$ iff $\m \dreach \n$ in $\pn'$.

  \item \emph{Rules $R_4$ and $R_5$.}\ We have $P'' = \set{p}$. For
    all $\m'$ and $\n'$, $\m' \dreach \n'$ in $\pn'$ iff
    $\pi_0(\m') \dreach \pi_0(\n')$ in $\pn$. Further, for all $t \in
    T$ and $p' \in P'$: either $\pre{t}[p] = 1$ implies $\pre{t}[p'] =
    0$; or $\post{t}[p] = 1$ implies $\post{t}[p'] = 0$. Also, for
    $\D \neq \Z$, if $\exists \m
    : \imarked{1} \dreach \m \not \dreach \fmarked{1}$ holds in $\pn$,
    then $\exists \m'
    : \imarked{1} \dreach \m' \not \dreach \fmarked{1}$ holds in
    $\pn'$.


  \item \emph{Rule $R_6$.}\ We have $P'' = \set{p_2, \ldots,
    p_k}$. There exists $p_1 \in P'$ such that for all $\n \in
    P^{\D}$, if $\sum_{i = 1}^k \m[p_i] = \sum_{i=1}^k\n[p_i]$ and
    $\n[p'] = \m[p']$ for $p' \in P' \setminus \set{p_1}$, then $\m
    \dreach \n$. Moreover, if $\m[p_i] = \n[p_i] = 0$ for $i > 1$,
    then $\m \dreach \n$ in $\pn$ iff $\pi(\m) \dreach \pi(\n)$ in
    $\pn'$.
  \end{itemize}
\end{restatable}

\section{Using relaxations for generalised soundness}
\label{sec:gen:sound}
In this section, we explain how reachability relaxations can be
leveraged in order to semi-decide generalised soundness of workflow
nets. More precisely, we state two necessary conditions for a workflow
net to be generalised sound: one phrased in terms of integer
reachability, and one in terms of continuous reachability. Furthermore,
for each condition we: (1)~show that it is preserved under reduction
rules, and (2)~establish its computational complexity. Overall, this
means that to conclude that a given workflow net $\pn$ is \emph{not}
generalised sound, one may first reduce $\pn$, and \emph{then}
efficiently test for one of these two necessary conditions.

For integer boundedness, we need the mild assumption of
nonredundancy. Let $\pn = (P, T, F)$ be a workflow net. We say that a
place $p \in P$ is \emph{nonredundant}\footnote{This notion is adapted
from batch workflow nets considered in~\cite{HSV04}.} if there exist
$k \in \Nn$ and $\m \in \N^P$ such that $\imarked{k} \reach \m$ and
$\m[p] \geq 1$. It is known (and simple to see) that redundant places
can be removed from a workflow net without changing whether it is
generalised sound. Moreover, testing whether a place is nonredundant
can be done in polynomial time. Indeed,
by \Cref{lem:cont-equiv-reach}, it amounts to testing for the
existence of some $\m \in \Qpos^P$ such that $\imarked{1} \creach \m$
and $\m[p] > 0$. The latter is known as a \emph{coverability} query
and it can be checked in polynomial time~\cite{FH15}. Thus, in order
to test whether a given workflow net is generalised sound, one can
first remove its redundant places. We call a workflow net without
redundant places a \emph{nonredundant workflow net}.

\subsection{Integer unboundedness}

Recall that a marked Petri net $(\pn, \m)$ is \emph{bounded} if there
exists $b \in \N$ such that $\m' \in \Reach{\pn, \m}$ implies $\m'
\leq \vec{b}$. It is well-known that any $1$-sound workflow net must
be bounded from $\imarked{1}$~\cite{AL97}. In particular, this means
that boundedness is a necessary condition for generalised
soundness. However, testing boundedness has extensive computational
cost as it is EXPSPACE-complete~\cite{CLM76,Rac78}. Consider the
relaxed property of \emph{integer boundedness}. It is defined as
boundedness, but where ``$\m' \in \Reach{\pn,\m}$'' is replaced with
``$\m' \in \ZReach{\pn, \m} \cap \N^P$''.

\begin{proposition}[{\cite[Lemma~5.9]{BMO22}}]\label{prop:zunbounded-notgen}
  Let $\pn$ be a nonredundant workflow net. If $(\pn, \imarked{1})$ is
  integer unbounded, then $\pn$ is not generalised sound.
\end{proposition}

\begin{restatable}{proposition}{propIntBoundRed}\label{prop:int:bound:red}
  The reduction rules from~\cite{bride2017reduction} preserve integer
  unboundedness.
\end{restatable}

Next, we establish the complexity of integer unboundedness in two steps. The
first step, in the next proposition, shows that testing integer
boundedness amounts to a simple condition, independent of the initial
marking. The second step shows the condition can be translated
into a linear program over $\Q$, rather than $\N$. As a corollary,
integer unboundedness is testable in polynomial time.

\begin{proposition}\label{prop:zunbounded:charac}
  A marked Petri net $(\pn, \m)$ is integer unbounded iff there exists
  a marking $\m' > \vec{0}$ such that $\0 \zreach \m'$ (independent of $\m$).
\end{proposition}

\begin{proof}
  Let $\pn = (P, F, T)$ be a Petri net and let $\m \in \N^P$.
  
  $\Rightarrow$) By assumption, there exist $\m_0, \m_1, \ldots \in
  \ZReach{\pn, \m} \cap \N^P$ such that, for every $i \in \N$, it is
  the case that $m_i \not\leq \vec{i}$. Since $(\N^P, \leq)$ is
  well-quasi-ordered, there exist indices $i_0, i_1, \ldots$ such that
  $\m_{i_j} \leq \m_{i_k}$ for all $j < k$. Without loss of
  generality, we can assume that $\m_{i_j} < \m_{i_k}$ for all $j <
  k$, as we could otherwise extract such a subsequence. Recall that
  each $\m_{i_\ell} \in \ZReach{\pn, \m}$. Let $\pi_\ell \in T^*$ be
  such that $\m \ztrans{\pi_\ell} \m_{i_\ell}$. Let $\vec{x}_\ell \in
  \N^T$ be the vector such that $\vec{x}_\ell(t)$ indicates the number
  of occurrences of transition $t$ in $\pi_\ell$. Since $(\N^T, \leq)$
  is well-quasi-ordered, there exist $j < k$ such that $\vec{x}_j \leq
  \vec{x}_k$. Let $\m' \defeq \m_{i_k} - \m_{i_j}$ and $\pi \defeq
  \prod_{t \in T} t^{(\vec{x}_k[t] - \vec{x}_\ell[t])}$. We have
  $\vec{0} \ztrans{\pi} \m' > \vec{0}$ as desired since:
  \begin{multline*}
    \m'
    = \m_{i_k} - \m_{i_j}
    = (\m + \effect{\pi_k}) - (\m + \effect{\pi_\ell})
    = \effect{\pi_k} - \effect{\pi_\ell} \\
    = \sum_{t \in T} \vec{x}_k[t] \cdot \effect{t} -
       \sum_{t \in T} \vec{x}_\ell[t] \cdot \effect{t}
    = \sum_{t \in T} (\vec{x}_k - \vec{x}_\ell)[t] \cdot \effect{t}
    = \effect{\pi}.
  \end{multline*}
  
  $\Leftarrow$) By assumption $\vec{0} \ztrans{\pi} \m' > \vec{0}$. In
  particular, this means that $\m \ztrans{\pi} \m + \m' \ztrans{\pi}
  \m + 2\m' \ztrans{} \cdots$. Therefore, $(\pn, \m)$ is not integer
  bounded. \qed
\end{proof}

\begin{restatable}{proposition}{propIntBoundLin}\label{prop:int:bound:lin}
  A marked Petri net $(\pn, \m)$, where $\pn = (P, T, F)$, is integer
  unbounded iff this system has a solution:
  $\exists \vec{x} \in \Qpos^T : \sum_{t \in
  T} \vec{x}[t] \cdot \effect{t} > \vec{0}$. In particular, given a
  workflow net $\pn$, testing integer boundedness of
  $(\pn, \imarked{1})$ can be done in polynomial time.
\end{restatable}

\subsection{Continuous soundness}

Let us now introduce a continuous variant of $1$-soundness based on
continuous reachability. We prove that this variant, which we call
\emph{continuous soundness}, is a necessary condition for generalised
soundness, and preserved by reduction rules. Moreover, we show that
continuous soundness is coNP-complete, and relates to integer
boundedness.

We say that a workflow net $\pn$ is \emph{continuously sound} if for
all continuous markings $\m \in \QReach{\pn, \imarked{1}}$ it is the
case that $\m \creach \fmarked{1}$.

\begin{theorem}\label{thm:cont-necessary}
  Continuous unsoundness implies generalised unsoundness.
\end{theorem}

\begin{proof}
  Let $\pn = (P, T, F)$ be a workflow net that is not continuously
  sound. By definition of continuous soundness, there exists some
  continuous marking $\m \in \Qpos^P$ such that
  $\imarked{1} \creach \m$ and
  $\m \not\creach \fmarked{1}$. By \Cref{lem:cont-equiv-reach}, there
  exists $b \in \Nn$ such that $\imarked{b} \reach
  b \cdot \m$. Furthermore, by \Cref{lem:cont-equiv-reach},
  $b \cdot \m \not \reach \fmarked{b}$. This means that $\pn$ is not
  $b$-sound, and consequently not generalised sound. \qed
\end{proof}

\begin{restatable}{proposition}{propContBoundRed}\label{prop:int:cont:red}
  The reduction rules from~\cite{bride2017reduction} preserve
  continuous soundness.
\end{restatable}


\begin{restatable}{theorem}{thmCoNP}\label{thm:conp}
  Continuous soundness is coNP-complete. Moreover, coNP-hardness holds
  even if the underlying graph of the given workflow net is acyclic.
\end{restatable}

\begin{proof}[of membership in coNP]
  The \emph{inclusion problem} consists in determining whether, given
  Petri nets $\pn$ and $\pn'$ over a common set of places, and
  markings $\m$ and $\m'$, it is the case that $\QReach{\pn, \m}
  \subseteq \QReach{\pn', \m'}$. The inclusion problem is known to be
  coNP-complete~\cite[Prop.~4.6]{BFHH17}.

  Let $\pn = (P, T)$ be a workflow net. Let $\pn^{-1} = (P, T^{-1})$
  be defined as $\pn$ but with its transitions reversed, \ie\ where
  $T^{-1} \defeq \{t^{-1} \mid t \in T\}$ with $\pre{(t^{-1})} \defeq
  \post{t}$ and $\post{(t^{-1})} \defeq \pre{t}$. It is the case that
  $\m \creach \m'$ in $\pn$ iff $\m' \creach \m$ in
  $\pn^{-1}$. Observe that $\pn$ is continuously sound iff the
  following holds for all $\m$:
  \[
  \m \in \QReach{\pn, \imarked{1}} \implies
  \omarked{1} \in \QReach{\pn, \m}.
  \]
  So, as $\omarked{1} \in \QReach{\pn, \m}$ is equivalent to $\m \in
  \QReach{\pn^{-1}, \omarked{1}}$, continuous soundness holds iff
  $
  \QReach{\pn, \imarked{1}} \subseteq
  \QReach{\pn^{-1}, \omarked{1}}
  $.
  As inclusion can be tested in coNP, membership follows. \qed
\end{proof}

\begin{proof}[of coNP-hardness]
  We give a reduction from the problem of determining whether a
  Boolean formula in disjunctive normal form (DNF) is a tautology. We
  adapt a construction from~\cite{tiplea2015acyclic} used to show that
  soundness in acyclic workflow nets is coNP-hard. The proof is more
  challenging under the continuous semantics as several variable
  valuations and clauses can be simultaneously used.

  The reduction is depicted in \Cref{fig:cont-sound-conp} for $\varphi
  = (x_1 \land x_2 \land \neg x_4) \lor (\neg x_1 \land x_3 \land
  x_4)$. In general, let $\varphi = \bigvee_{j \in [1..k]} C_j$ be a
  Boolean formula in DNF with $k$ clauses over variables $x_1, \ldots,
  x_m$. We define a workflow net $\pn_\varphi = (P, T, F)$.

  \begin{figure}[h!]
    \centering
    \begin{tikzpicture}[auto, node distance=0.8cm, transform shape, scale=0.9]
      \colorlet{colC1}{blue}
      \colorlet{colC2}{magenta}
      \colorlet{colBack}{white}
      \colorlet{colFont}{black!75}
      
      \tikzstyle{cplace} = [place, font=\scriptsize, text=colFont,
                            minimum size=15pt];
      \tikzstyle{ctransition} = [transition, font=\scriptsize, text=colFont,
                                 minimum height=12pt, minimum width=15pt];
      
      \node[cplace] (i) {$\initial$};
      \node[ctransition, right of=i] (init) {$t_\text{init}$};
      
      \node[cplace, above right=0.25cm and 0.5cm of init] (x2n) {$p_{2,?}$};
      \node[cplace, above=of x2n] (x1n) {$p_{1,?}$};
      \node[cplace, below right=0.25cm and 0.5cm of init] (x3n) {$p_{3,?}$};
      \node[cplace, below=of x3n] (x4n) {$p_{4,?}$};

      \node[cplace, below=3.35cm of init, xshift=4.75cm] (clause)
           {$p_\text{cl}$};
      \node[ctransition, above  left=0.35cm and 0.5cm of clause,
            fill=colC1!40] (c1) {$c_1$};
      \node[ctransition, above right=0.35cm and 0.5cm of clause,
            fill=colC2!40] (c2) {$c_2$};

      \node[ctransition, above right=-0.1cm and 0.5cm of x1n] (tx1t) {$v_{1,1}$};
      \node[ctransition, below right=-0.1cm and 0.5cm of x1n] (tx1f) {$v_{1,0}$};
      \node[cplace, right=0.5cm of tx1t] (x1t) {$p_{1,1}$};
      \node[cplace, right=0.5cm of tx1f] (x1f) {$p_{1,0}$};

      \node[ctransition, above right=-0.1cm and 0.5cm of x2n] (tx2t) {$v_{2,1}$};
      \node[ctransition, below right=-0.1cm and 0.5cm of x2n] (tx2f) {$v_{2,0}$};
      \node[cplace, right=0.5cm of tx2t] (x2t) {$p_{2,1}$};
      \node[cplace, right=0.5cm of tx2f] (x2f) {$p_{2,0}$};

      \node[ctransition, above right=-0.1cm and 0.5cm of x3n] (tx3t) {$v_{3,1}$};
      \node[ctransition, below right=-0.1cm and 0.5cm of x3n] (tx3f) {$v_{3,0}$};
      \node[cplace, right=0.5cm of tx3t] (x3t) {$p_{3,1}$};
      \node[cplace, right=0.5cm of tx3f] (x3f) {$p_{3,0}$};

      \node[ctransition, above right=-0.1cm and 0.5cm of x4n] (tx4t) {$v_{4,1}$};
      \node[ctransition, below right=-0.1cm and 0.5cm of x4n] (tx4f) {$v_{4,0}$};
      \node[cplace, right=0.5cm of tx4t] (x4t) {$p_{4,1}$};
      \node[cplace, right=0.5cm of tx4f] (x4f) {$p_{4,0}$};

      \node[ctransition, right=5cm of tx1t, fill=colBack] (cx1t)
           {$\overline{v}_{1,1}$};
      \node[ctransition, right=5cm of tx1f, fill=colBack] (cx1f)
           {$\overline{v}_{1,0}$};
      \node[cplace, below  left=-0.1cm and 0.5cm of cx1t, fill=colBack] (cx1)
           {$q_1$};
      \node[cplace, below right=-0.1cm and 0.5cm of cx1t, fill=colBack] (cx1c)
           {$r_1$};

      \node[ctransition, right=5cm of tx2t, fill=colBack] (cx2t)
           {$\overline{v}_{2,1}$};
      \node[ctransition, right=5cm of tx2f, fill=colBack] (cx2f)
           {$\overline{v}_{2,0}$};
      \node[cplace, below  left=-0.1cm and 0.5cm of cx2t, fill=colBack] (cx2)
           {$q_2$};
      \node[cplace, below right=-0.1cm and 0.5cm of cx2t, fill=colBack] (cx2c)
           {$r_2$};

      \node[ctransition, right=5cm of tx3t, fill=colBack] (cx3t)
           {$\overline{v}_{3,1}$};
      \node[ctransition, right=5cm of tx3f, fill=colBack] (cx3f)
           {$\overline{v}_{3,0}$};
      \node[cplace, below  left=-0.1cm and 0.5cm of cx3t, fill=colBack] (cx3)
           {$q_3$};
      \node[cplace, below right=-0.1cm and 0.5cm of cx3t, fill=colBack] (cx3c)
           {$r_3$};

      \node[ctransition, right=5cm of tx4t, fill=colBack] (cx4t)
           {$\overline{v}_{4,1}$};
      \node[ctransition, right=5cm of tx4f, fill=colBack] (cx4f)
           {$\overline{v}_{4,0}$};
      \node[cplace, below  left=-0.1cm and 0.5cm of cx4t, fill=colBack] (cx4)
           {$q_4$};
      \node[cplace, below right=-0.1cm and 0.5cm of cx4t, fill=colBack] (cx4c)
           {$r_4$};

      \node[ctransition, below right of=cx2c] (end) {$t_\text{fin}$};
      \node[cplace, right of=end] (f) {$\output$};

      \path[->]
      (i) edge node {} (init)
      
      (init) edge[out=90,  in=235]  node {} (x1n)
      (init) edge[out=90,  in=180]  node {} (x2n)
      (init) edge[out=-90, in=-180] node {} (x3n)
      (init) edge[out=-90, in=-235] node {} (x4n)

      (init) edge[out=-90, in=180, looseness=1.75] node {} (clause)

      (x1n)  edge node {} (tx1t)
      (x1n)  edge node {} (tx1f)
      (tx1t) edge node {} (x1t)
      (tx1f) edge node {} (x1f)

      (x2n)  edge node {} (tx2t)
      (x2n)  edge node {} (tx2f)
      (tx2t) edge node {} (x2t)
      (tx2f) edge node {} (x2f)

      (x3n)  edge node {} (tx3t)
      (x3n)  edge node {} (tx3f)
      (tx3t) edge node {} (x3t)
      (tx3f) edge node {} (x3f)

      (x4n)  edge node {} (tx4t)
      (x4n)  edge node {} (tx4f)
      (tx4t) edge node {} (x4t)
      (tx4f) edge node {} (x4f)

      (cx1)  edge node {} (cx1t)
      (cx1)  edge node {} (cx1f)
      (cx1t) edge node {} (cx1c)
      (cx1f) edge node {} (cx1c)

      (cx2)  edge node {} (cx2t)
      (cx2)  edge node {} (cx2f)
      (cx2t) edge node {} (cx2c)
      (cx2f) edge node {} (cx2c)

      (cx3)  edge node {} (cx3t)
      (cx3)  edge node {} (cx3f)
      (cx3t) edge node {} (cx3c)
      (cx3f) edge node {} (cx3c)

      (cx4)  edge node {} (cx4t)
      (cx4)  edge node {} (cx4f)
      (cx4t) edge node {} (cx4c)
      (cx4f) edge node {} (cx4c)
      ;

      \path[->, gray]
      (x1t) edge node {} (cx1t)
      (x1f) edge node {} (cx1f)
      
      (x2t) edge node {} (cx2t)
      (x2f) edge node {} (cx2f)
      
      (x3t) edge node {} (cx3t)
      (x3f) edge node {} (cx3f)
      
      (x4t) edge node {} (cx4t)
      (x4f) edge node {} (cx4f)
      ;

      \path[->]
      (cx1c) edge[out=-45,in=90]  node {} (end)
      (cx2c) edge[out=0,  in=90]  node {} (end)
      (cx3c) edge[out=0,  in=-90] node {} (end)
      (cx4c) edge[out=45, in=-90] node {} (end)
      (end)  edge node {} (f)
      ;

      \begin{pgfonlayer}{back}
        \path[->, colC1, densely dotted]
        (clause) edge[out=135, in=-90] node {} (c1)
        
        (x1t) edge[out=-45, in=110] node {} (c1)
        (x2t) edge[out=-55, in=110] node {} (c1)
        (x4f) edge[out=-45, in=110] node {} (c1)
        
        (c1) edge[out=70, in=-125] node {} (cx1c)
        (c1) edge[out=70, in=-90]  node {} (cx2c)
        (c1) edge[out=65, in=180]  node {} (cx3)
        (c1) edge[out=70, in=-115] node {} (cx4c)
        ;
        
        \path[->, colC2, dashed]
        (clause) edge[out=45, in=-90] node {} (c2)
        
        (x1t) edge[out=-30, in=110] node {} (c2)
        (x3t) edge[out=-55, in=110] node {} (c2)
        (x4t) edge[out=-45, in=110] node {} (c2)
        
        (c2) edge[out=70, in=-105] node {} (cx1c)
        (c2) edge[out=70, in=225]  node {} (cx2)
        (c2) edge[out=70, in=-90]  node {} (cx3c)
        (c2) edge[out=70, in=-80]  node {} (cx4c)
        ;
      \end{pgfonlayer}
    \end{tikzpicture}

    \caption{A workflow net $\pn_\varphi$ such that $\pn_\varphi$ is
      continuously sound iff $\varphi = (x_1 \land x_2 \land \neg x_4)
      \lor (x_1 \land x_3 \land x_4)$ is a tautology. Places and
      transitions contain their names (not values). Arcs corresponding
      to the first and second clauses are respectively dotted and
      dashed.}\label{fig:cont-sound-conp}
\end{figure}
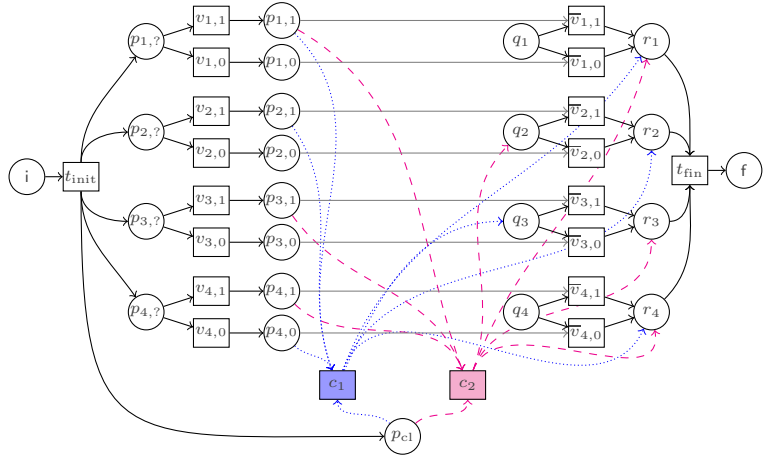

  \emph{Definition.} The places are defined as $P \defeq \{\initial,
  p_\text{cl}, \output\} \cup P_\text{var} \cup P_\text{clean}$, where
  $
  P_\text{var}
  \defeq \bigcup_{i \in [1..m]} \{p_{i,?}, p_{i,1}, p_{i,0}\}
  \text{ and }
  P_\text{clean}
  \defeq \bigcup_{i \in [1..m]} \{q_i, r_i\}
  $.
  The transitions are defined as $T \defeq \{t_\text{init},
  t_\text{fin}\} \cup T_\text{var} \cup T_\text{clauses} \cup
  T_{\overline{\text{var}}}$, where
  \[
  T_\text{var}
  \defeq \bigcup_{\mathclap{i \in [1..m]}}\ \{v_{i,1}, v_{i,0}\},
  T_\text{clauses}
  \defeq \{c_i \mid i \in [1..k]\}
  \text{ and }
  T_{\overline{\text{var}}}
  \defeq \bigcup_{\mathclap{i \in [1..m]}}\
  \{\overline{v}_{i,1}, \overline{v}_{i,0}\}.
  \]

  Let us explain how $\pn_\varphi$ is \emph{intended} to
  work. Transition $t_\text{init}$ enables the initialization of
  variables and the selection of a clause that satisfies $\varphi$,
  i.e.\ $\pre{t_\text{init}} \defeq \imarked{1}$ and
  $\post{t_\text{init}} \defeq \{p_{i,?} \colon 1 \mid i \in [1..m]\}
  + \{p_\text{cl} \colon 1\}$. A token in place $p_{i,b}$ indicates
  that variable $x_i$ has been assigned value $b$ (where ``$?$''
  indicates ``none''). Consequently, we have $\pre{v_{i,b}} \defeq
  p_{i,?}$ and $\post{v_{i,b}} \defeq p_{i,b}$ for each $i \in [1..m]$
  and $b \in \{0, 1\}$.

  Transition $c_j$ consumes a token associated to each literal of
  clause $C_j$, \ie\ $\pre{c_j} \defeq \{v_{i,1} \mid x_i \in C_j\} +
  \{v_{i,0} \mid \neg x_i \in C_j\}$. A token in place $q_i$ indicates
  that variable $x_i$ is not needed anymore (due to some satisfied
  clause). A token in place $r_i$ indicates that variable $x_i$ has
  been discarded. Therefore, transition $c_j$ produces these tokens:
  $
  \post{c_j} \defeq
  \{q_i \mid x_i \notin C_j \land \neg x_i \notin C_j\} + 
  \{r_i \mid x_i \in C_j \lor \neg x_i \in C_j\}
  $.

  Transition $\overline{v}_{i,b}$ discards variable $x_i$,
  i.e.\ $\pre{\overline{v}_{i,b}} \defeq \{p_{i,b}, q_i\}$ and
  $\pre{\overline{v}_{i,b}} \defeq \{q_i\}$. Once each variable is
  discarded, transition $t_\text{fin}$ terminates the execution,
  i.e.\ $\pre{t_\text{fin}} \defeq \{r_i \mid i \in [1..m]\}$ and
  $\post{t_\text{fin}} \defeq \fmarked{1}$.

  \emph{Correctness.} Note that under $\creach$, the workflow net
  needs not to proceed as described. Indeed, it could, e.g., assign
  half a token to $p_{i,0}$ and half a token to $p_{i,1}$. Similarly,
  several clauses can be used, with distinct scaling
  factors. Nonetheless, $\pn_{\varphi}$ is continuously sound iff
  $\varphi$ is a tautology.

  $\Rightarrow$) Let $b_1, \ldots, b_m \in \{0, 1\}$. Let $\pi \defeq
  t_\text{init} v_{1, b_1} \cdots v_{m, b_m}$. We have:
  $
  \imarked{1}
  \ctrans{\pi}
  \{v_{i, b_i} \colon 1 \mid i \in [1..m]\} + \{p_\text{cl} \colon 1\}
  $.
  Since $\pn_\varphi$ is continuously sound by assumption, there must
  exists some $j \in [1..k]$ such that $c_j$ is enabled. This implies
  that clause $C_j$ is satisfied by the assignment. Hence, $\varphi$
  is a tautology.

  $\Leftarrow$) The proof is technical and involves several invariants (see
  appendix). \qed
\end{proof}

We may now prove that any nonredundant workflow net that is integer
unbounded is also continuously unsound (the reverse is not necessarily
true). Therefore, integer unboundedness relates to continuous
soundness much like continuous unsoundness relates to generalised
soundness.

\begin{proposition}\label{prop:zbound:continuous}
  Let $\pn$ be a nonredundant workflow net and $\m \in \N^P$. If
  $(\pn, \m)$ is integer unbounded, then $\pn$ is not continuously
  sound.
\end{proposition}

\begin{proof}
  Let $\pn = (P, T, F)$ and $\m \in \N^P$ be such that $(\pn, \m)$ is
  not integer bounded. By \Cref{prop:zunbounded:charac}, there exists
  $\m' > \vec{0}$ such that $\vec{0} \zreach \m'$. By nonredundancy, there
  exist $\lambda \in \Nn$ and $\m'' \in \N^P$ such that $\imarked{\lambda}
  \reach \fmarked{1} + \m''$.

  In \cite[Lemma~12]{HSV04}, it is shown that $\imarked{k} \zreach \n$
  implies the existence of some $\ell \in \N$ such that $\imarked{k +
    \ell} \reach \fmarked{\ell} + \n$. By invoking this lemma with $k
  \defeq 0$ and $\n \defeq \m'$, we obtain $\imarked{\ell} \reach
  \fmarked{\ell} + \m'$ for some $\ell \in \N$.

  Altogether, $\imarked{\lambda + \ell} \reach \fmarked{\lambda +
    \ell} + \m' + \m''$. Since $\lambda + \ell \geq 1$,
  \Cref{lem:cont-equiv-reach} yields $\imarked{1} \creach \fmarked{1}
  + \m'''$ where $\m''' \defeq (1 / (\lambda + \ell)) \m'$. As every
  transition of a workflow net produces at least one token, this
  contradicts the fact that $\pn$ is continuously sound. Indeed, it is
  impossible to fully get rid of $\m''' > \vec{0}$. \qed
\end{proof}


\section{Using relaxations for structural soundness}
\label{sec:struct:sound}

A workflow net $\pn$ is \emph{$k$-quasi-sound} if $\imarked{k}
  \reach \fmarked{k}$. Furthermore, $\pn$ is \emph{structurally
  quasi-sound} if it is $k$-quasi-sound for some $k \in \Nn$.

As observed in~\cite{ctiplea2005structural}, structural
quasi-soundness is a necessary condition for structural soundness. The
notion of structural quasi-soundness is naturally generalised to an
arbitrary Petri net $\pn = (P, T, F)$. Given markings $\m, \m' \in
  \N^P$, we say that $\m$ \emph{structurally reaches} $\m'$ in $\pn$ if
$k \cdot \m \reach k \cdot \m'$ for some $k \in \Nn$. A workflow net
is structurally quasi-sound iff $\vec{m} \defeq \imarked{1}$
structurally reaches $\vec{m}' \defeq \omarked{1}$. So, the
observation of~\cite{ctiplea2005structural} can be rephrased as
follows.

\begin{proposition}
  Let $\pn$ be a workflow net. If $\imarked{1}$ does not structurally
  reach $\fmarked{1}$ in $\pn$, then $\pn$ is not structurally sound.
\end{proposition}

The problem of structural quasi-soundness can be reduced to an
instance of the Petri net reachability
problem~\cite[Lemma~2.1]{ctiplea2005structural}. Intuitively, the
reduction produces a Petri net that nondeterministically chooses
multiples of $\imarked{1}$ and $\omarked{1}$ for which to check
reachability. Such an approach has a prohibitive computational cost as
Petri net reachability is Ackermann-complete. However, we
observe that structural reachability, and hence structural
quasi-soundness, is equivalent to continuous reachability by
\Cref{lem:cont-equiv-reach}.

\begin{proposition}\label{prop:struct:reach:cont}
  Let $\pn = (P, T, F)$ be a Petri net, and let $\m, \m' \in \N^P$ be
  markings. It is the case that $\m$ structurally reaches $\m'$ iff
  $\m \creach \m'$.
\end{proposition}




For a workflow net $\pn = (P, T, F)$, let $k_{\pn} \in \Nn \cup
\{\infty\}$ be the smallest number for which $\pn$ is
$k_{\pn}$-quasi-sound.
Then $\pn$ is structurally sound iff
$k_{\pn} \neq \infty$ and $\pn$ is $k_{\pn}$-sound~\cite[Thm~2.1]{ctiplea2005structural}. By
\Cref{prop:struct:reach:cont}, $k_{\pn} \neq \infty$
can be checked in polynomial time via a continuous reachability query.
Moreover, a lower bound on $k_{\pn}$ can be
obtained by computing $k_{\pn,\Z} \in \Nn \cup \{\infty\}$, defined as
the smallest value such that $\imarked{k} \zreach \fmarked{k}$.
We obtain a
better bound by defining $k_{\pn, \Qpos} \in \Nn \cup
\{\infty\}$ as the smallest value for which there is a continuous run
$\pi = \lambda_1 t_1 \cdots \lambda_n t_n$ such that $\imarked{k}
\ctrans{\pi} \fmarked{k}$ and $\vec{\pi} \in \N^T$, where
$\vec{\pi}[t] \defeq \sum_{i \in [1..n] : t_i = t} \lambda_i$. Values
$k_{\pn,\Z}$ and $k_{\pn,\Qpos}$ can respectively be computed by
a translation to integer linear programming, and a decidable
optimization modulo theory.

\begin{restatable}{proposition}{propIntegerK}\label{prop:integerk}
  Let $\pn$ be a workflow net. It is the case that $k_{\pn,\Z} \leq
  k_{\pn,\Qpos} \leq k_{\pn}$. Moreover, $k_{\pn,\Z}$ can be computed
  from an integer linear program $\mathcal{P}$; $k_{\pn,\Qpos}$ can be obtained
  by computing $\min k \in \Nn : \varphi(k)$ where $\varphi$ is a
  formula from the existential fragment of mixed linear arithmetic
  $\varphi$, \ie\ $\exists \mathsf{FO}(\Q, \Z, <, +)$; and both $\mathcal{P}$
  and $\varphi$ are constructible in polynomial time from $\pn$.
\end{restatable}

\section{Free-choice workflow nets}
\label{sec:freechoice}
Let $\pn = (P, T, F)$ be a Petri net. We say that $\pn$ is
\emph{free-choice} if for any $s, t \in T$, it is the case that either
$\support{\pre{s}} \cap \support{\pre{t}} = \emptyset$ or $\pre{s} =
\pre{t}$. For example, the nets $\pn_\text{left}$ and
$\pn_\text{right}$ from \Cref{fig:pn} are respectively free-choice and
not free-choice.


It is known that generalised soundness is equivalent to $1$-soundness
in free-choice workflow nets~\cite{ping2004on}. We will show that the
same holds for structural soundness, and that, surprisingly, for
continuous soundness as well. This means that notions of soundness
collapse for free-choice nets. This is proven in the
forthcoming \Cref{lem:all:three}
and \Cref{thm:fc:equiv}, which form one of the main theoretical
contributions of this work.

Let $(\pn, \m)$ be a marked Petri net. We say that a transition $t$ is
\emph{quasi-live} in $(\pn, \m)$ if there exists $\m'$ such that $\m
\reach \m' \trans{t}$. Similarly, we say that a transition $t$ is
\emph{live} in $(\pn, \m)$ if for all $\m'$ such that $\m \reach \m'$,
$t$ is quasi-live in $(\pn,\m')$. In words, quasi-liveness states that
there is at least one way to enable $t$, and liveness states that $t$
can always be re-enabled. The set of \emph{quasi-live} and \emph{live}
transitions of $(\pn, \m)$ are defined respectively as $F(m) \defeq
\{t \in T \mid \text{$t$ is quasi-live in $(\pn, \m)$}\}$ and $L(m)
\defeq \{t \in T \mid \text{$t$ is live in $(\pn, \m)$}\}$.

\begin{restatable}{lemma}{lemAllThree}\label{lem:all:three}
  Let $\pn = (P, T, F)$ be a free-choice Petri net, let $c \in
  \Nn$, and let $\m \in \N^P$. The following statements hold.
  \begin{enumerate}
    \item There exists a marking $\m'$ such that $\m \reach \m'$ and
    $L(\m') = F(\m')$.\label{itm:all:three:a}
    
    \item If $L(\m) = F(\m)$, then $L(c \cdot \m) = F(c \cdot \m) =
    F(\m)$.\label{itm:all:three:b}
    
    \item If $L(c \cdot \m) = F(c \cdot \m)$, $c \cdot \m \reach
    \fmarked{c}$ and $(\pn, c \cdot \m)$ is bounded, then $\m =
    \fmarked{1}$.\label{itm:all:three:c}
  \end{enumerate}
\end{restatable}

\begin{restatable}{lemma}{lemContsoundBounded}\label{lem:contsoundbounded}
  Let $\pn$ be a workflow net. If $\pn$ is continuously sound,
  then $(\pn, \imarked{k})$ is bounded for all $k \in \Nn$.
\end{restatable}


\begin{theorem}\label{thm:fc:equiv}
  Let $\pn$ be a free-choice workflow net. These statements are
  equivalent: (1)~$\pn$ is $1$-sound, (2)~$\pn$ is generalised sound,
  (3)~$\pn$ is structurally sound, and (4)~$\pn$ is continuously
  sound.
\end{theorem}

\begin{proof}
  (1)~$\Rightarrow$~(2). This was shown in~\cite{ping2004on}.

  (2)~$\Rightarrow$~(3). By definition, if $\pn$ is $k$-sound for all
  $k$, then it is for some $k$.

%


  (2)~$\Rightarrow$~(4). By \Cref{thm:cont-necessary}.
  
%

  (3)~$\Rightarrow$~(1). Let $k \in \Nn$ be such that $\pn$ is
  $k$-sound. Let $\m \in \N^P$ be such that $\imarked{1} \reach
  \m$. By \Cref{lem:all:three}\eqref{itm:all:three:a}, there is a
  marking $\m' \in \N^P$ such that $\m \reach \m'$ and $F(\m') =
  L(\m')$. By \Cref{lem:all:three}\eqref{itm:all:three:b}, we have
  $L(k \cdot \m') = F(k \cdot \m') = F(\m')$.

  By $k$-soundness, $(\pn, \imarked{k})$ must be
  bounded~\cite[Proposition~3.2 and Lemma~3.6]{BMO22}. Thus, since
  $\imarked{k} \reach k \cdot \m \reach k \cdot \m'$, it is also the
  case that $(\pn, k \cdot \m')$ is bounded. By $k$-soundness, $k
  \cdot \m' \reach \fmarked{k}$. By invoking
  \Cref{lem:all:three}\eqref{itm:all:three:c} with $c \defeq k$, we
  conclude that $\m' = \fmarked{1}$. So, $\pn$ is $1$-sound as
  $\imarked{1} \reach \m \reach \m' = \fmarked{1}$.

  
  (4)~$\Rightarrow$~(1). Assume that $\pn$ is continuously sound. Let
  $\m \in \N^P$ be a marking such that $\imarked{1} \reach \m$. By
  \Cref{lem:all:three}\eqref{itm:all:three:a}, there exists $\m' \in
  \N^P$ such that $\m \reach \m'$ and $L(m') = F(m')$. Clearly,
  $\imarked{1} \creach \m'$ and by continuous soundness $\m' \creach
  \fmarked{1}$. By \Cref{lem:cont-equiv-reach}, there exists $b \in
  \Nn$ such that $b \cdot \m' \reach \fmarked{b}$.
  

  By \Cref{lem:contsoundbounded}, continuous soundness of $\pn$ implies that $(\pn, b \cdot
  \m')$ is bounded, as $\imarked{b} \reach b \cdot \m'$. Since $L(m') = F(m')$, it follows from
  \Cref{lem:all:three}\eqref{itm:all:three:b} that $L(b \cdot \m') =
  F(b \cdot \m')$. By invoking
  \Cref{lem:all:three}\eqref{itm:all:three:c} with $c \defeq b$, we
  derive $\m' = \fmarked{1}$. Therefore, $\pn$ is $1$-sound as
  $\imarked{1} \reach \m \reach \m' = \fmarked{1}$. \qed
\end{proof}

\section{Experimental evaluation}
\label{sec:experimental}
We implemented our approaches for generalised and structural soundness
in \csharp.\footnote{In the case of acceptance, we will
    submit an artifact to the artifact evaluation.} We test
continuous soundness via SMT solving. More precisely, we use an
existential $\psi_{\pn}$ formula of linear arithmetic,
i.e. $\mathsf{FO}(\Q, <, +)$, from~\cite{BFHH17}. This formula is such
that $\psi(\m, \m')$ holds iff $\m \creach \m'$ in $\pn$. Continuous
soundness amounts to the $\exists \forall$-formula
$\psi_{\pn}(\imarked{1}, \m) \land \neg \psi_{\pn}(\m,
    \fmarked{1})$. To solve such formulas, we use
Z3~\cite{moura2008z3}. We further use Z3 to decide structural
quasi-soundness and compute $k_{\pn,\Qpos}$ (see
\Cref{prop:integerk}), again via the formulas of~\cite{BFHH17}.

We evaluated our prototype implementation on a standard benchmark
suite used regularly in the literature, and a novel suite of synthetic
instances where generalised or structural soundness are hard to decide
with a naive approach.

We compared with two established tools for soundness: LoLA
(v2.0)~\cite{lola}, and Woflan~\cite{verbeek2000woflan}.\footnote{A
    version of Woflan suitable for running without user interaction was
    provided, via personal communication, by its maintainer.} The
latter can only decide \emph{classical} soundness ($1$-soundness +
quasi-liveness). Nonetheless, we use quasi-live instances, so for
which $1$-soundness and classical soundness are equivalent. We further
use a  transformation to reduce the verification of
$k$-soundness to the one of $1$-soundness~\cite[Lemma~3.6]{BMO22}. On
the other hand, LoLA can directly decide $k$-soundness. To do so, we start from
$\imarked{k}$ and check a CTL formula of the form $\forall
    \mathsf{G}\, \exists \mathsf{F}\, ((\m[\output] = k) \land
    \bigwedge_{p \neq \output} \m[p] = 0)$.


Experiments were run on an 8-Core
Intel\textregistered~Core\texttrademark~i7-7700 CPU @ 3.60GHz with
Ubuntu~18.04. We limited memory to $\sim$8GB\@, and time to $120$s for
each instance. Tools were called from a Python script. For LoLA and
our implementation, we used the \emph{time} module to measure
time. Running Woflan involves some overhead, so we instead take the
total verification time reported by Woflan itself.

\subsection{Free-choice benchmark suite}

The benchmark suite encompasses 1386 free-choice Petri nets that
represent business processes modeled in the IBM WebSphere Business
Modeler. It was originally presented
in~\cite{fahland2009instantaneous}, and has been studied frequently in
the literature~\cite{bride2017reduction,favre2016diagnostic}. These
nets are not workflow nets by our definition, but can be transformed
using a known
procedure~\cite{kiepuszewski2003fundamentals}. Intuitively, the nets
are workflow nets with multiple final places, and the procedure adds a
dedicated output place and ensures that the resulting workflow net
represents the desired behaviour. However, roughly $1\%$ of the nets
are not workflow nets by our definition even after the procedure, as
they contain nodes that are not on a path from $\initial$ to
$\output$. We removed these nets.

We further checked each net for safety using LoLA and dropped unsafe
nets. Recall that $(\pn, \imarked{1})$ is sound if each reachable
marking has at most one token per place. Unsafe instances can be
dropped as unsafety implies $1$-unsoundness in free-choice
nets~\cite[Thm.~4.2 and~4.4]{verbeek2001diagnosing}, and as existing
methods for checking safety, \eg via state-space exploration with
partial order reductions, are very efficient (here needing a mean of
$3$ms). Thus, we considered safe instances only. Among the 1386
instances, 1382 are workflow nets, and 977 are further safe.


We also invoked an implementation of the reduction rules
of~\cite{bride2017reduction} to reduce the size of all
instances.\footnote{At time of writing, an implementation is available
    at \url{https://github.com/LoW12/Hadara-AdSimul}.} As discussed in the
introduction, the rules can reduce some instances to trivially sound
nets. However, even the size of
nontrivial reduced instances tends to be small, with an average number
of places and transitions of roughly $14$, while three quarters of
nets have at most $18$ places and transitions. This is small enough
that a complete state-splace enumeration is often
feasible,
in particular as the nets are safe and especially LoLA
utilizes powerful partial order reductions for such nets. As we want to focus
on scalability, we chained instances to produce challenging synthetic nets based on real-world
instances.
This is a natural way of constructing workflow nets, intuitively, the final process can be composed of many subtasks.
It can be seen as a special case of refinement operations, studied in the context of generalised soundness~\cite{van2003soundness}.

The chaining procedure merges two workflow nets $\pn = (P, T, F)$ and
$\pn' = (P', T', F')$ into $\pn'' \defeq (P'', T'', F'')$ where
$P'' \defeq P \cup P'$, $T'' \defeq T \cup T' \cup \{t_{aux}\}$ with
$F''$ as $F' + F''$ extended with
$\pre{t_{\mathrm{aux}}}[\output] \defeq 1$,
$\post{t_{\mathrm{aux}}}[\initial'] \defeq 1$, and
$\pre{t_{\mathrm{aux}}}[p] = \post{t_{\mathrm{aux}}}[p'] \defeq 0$ for
other entries. It is readily seen that this construction (1)~produces
a free-choice net if both $\pn$ and $\pn'$ are free-choice; and
(2)~preserves safety.

This way, we generated large instances by using $\ell \in \{1, 21,
    41, \dots, 401\}$ randomly chosen unreduced safe instances from the
benchmark suite as inputs to be chained into one instance, then
reduced that instance.
For each number $\ell$, we produced 20 combined nets, with a fresh
random choice each time, in order to have a more representative
collection of nets for $\ell$. This resulted in 420 instances, of
which 405 are nontrivial after applying reduction rules.

A caveat is that such large nets may seem unlikely to arise in
practice. It seems a human designer would avoid
designing highly complex processes corresponding to Petri nets
with thousands of places.  However, process models are not only
explicitly written by humans, but also machine-generated, \eg by
mining event logs~(see \cite{van2016data} for a book on the topic). In particular, being
free-choice is preserved by chaining, so a large free-choice net may
``hide'' and combine several less complex processes, which might
necessitate analyzing large workflow nets.

\subsubsection{Results.}

We checked the safe free-choice instances obtained as explained above
for $1$-soundness using LoLA, Woflan and
our implementation of continuous soundness. The results are shown on
the left of \Cref{fig:fc-results}. The right-hand side of the figure
provides an overview over the sizes of the nets. In each case, $N$
refers to the number of original instances that were chained to create
each instance. 

\begin{figure}
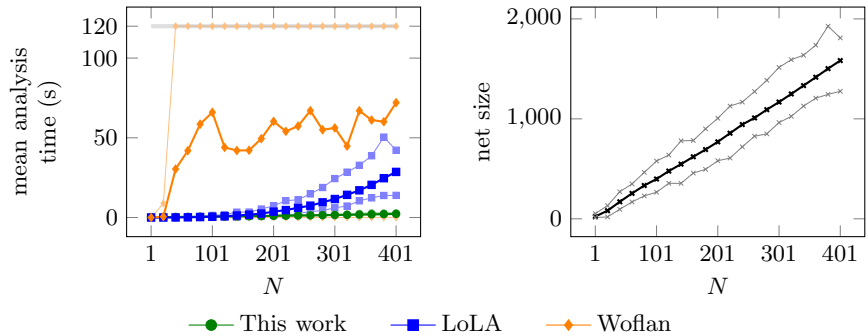

    \centering
    \plotcanvaswithextraargs{$N$}{\begin{tabular}{c} mean analysis \\ time (s)\end{tabular}}{
        \addplot[ultra thin, color=colConti!50!white, mark=*, mark size=1.2pt] coordinates {(1,0.24875164031982422) (21,0.25934910774230957) (41,0.2969787120819092) (61,0.34603285789489746) (81,0.38550233840942383) (101,0.41761255264282227) (121,0.5175187587738037) (141,0.5033724308013916) (161,0.6341457366943359) (181,0.6340017318725586) (201,0.7953159809112549) (221,0.781745433807373) (241,0.9047155380249023) (261,1.0543396472930908) (281,1.1093416213989258) (301,1.2764415740966797) (321,1.3729355335235596) (341,1.5636298656463623) (361,1.6404368877410889) (381,1.8053691387176514) (401,1.8820371627807617)};

        \addplot[ultra thin, color=colLola!50!white, mark=square*, mark size=1.2pt] coordinates {(1,0.0032896995544433594) (21,0.003237485885620117) (41,0.009343385696411133) (61,0.03467702865600586) (81,0.08178305625915527) (101,0.12643218040466309) (121,0.3325989246368408) (141,0.2991626262664795) (161,0.7066307067871094) (181,0.8594794273376465) (201,1.317361831665039) (221,1.4517436027526855) (241,2.4387874603271484) (261,3.7813377380371094) (281,4.288681983947754) (301,6.137068271636963) (321,7.181021451950073) (341,10.475846767425537) (361,12.33228874206543) (381,13.807121515274048) (401,13.894974708557129)};

        \addplot[ultra thin, color=colWoflan!50!white, mark=diamond*, mark size=1.2pt] coordinates {(1,0.002) (21,0.002) (41,0.003) (61,0.003) (81,0.004) (101,0.005) (121,0.005) (141,0.006) (161,0.008) (181,0.009) (201,0.01) (221,0.01) (241,0.012) (261,0.014) (281,0.016) (301,0.018) (321,0.028) (341,0.026) (361,0.036) (381,0.04) (401,0.032)};

        \addplot[ultra thin, color=colConti!50!white, mark=*, mark size=1.2pt] coordinates {(1,0.27018189430236816) (21,0.3120253086090088) (41,0.4360194206237793) (61,0.5158348083496094) (81,0.6339051723480225) (101,0.833796501159668) (121,0.8922023773193359) (141,1.0395622253417969) (161,1.055588722229004) (181,1.217644214630127) (201,1.3503105640411377) (221,1.6982512474060059) (241,1.5905685424804688) (261,1.6604011058807373) (281,1.8535799980163574) (301,2.072716236114502) (321,2.148928642272949) (341,2.5639662742614746) (361,2.983848810195923) (381,2.8429923057556152) (401,2.9164958000183105)};

        \addplot[ultra thin, color=colLola!50!white, mark=square*, mark size=1.2pt] coordinates {(1,0.005075693130493164) (21,0.020010948181152344) (41,0.1506800651550293) (61,0.29684996604919434) (81,0.703944206237793) (101,1.3222572803497314) (121,1.83095121383667) (141,3.360896587371826) (161,3.420379161834717) (181,5.309957265853882) (201,7.344234228134155) (221,10.57127594947815) (241,11.159202098846436) (261,14.828051567077637) (281,18.81239151954651) (301,24.42066502571106) (321,28.465094327926636) (341,32.755656003952026) (361,38.74089980125427) (381,50.431984424591064) (401,42.217225074768066)};

        \addplot[ultra thin, color=colWoflan!50!white, mark=diamond*, mark size=1.2pt] coordinates {(1,0.028) (21,8.794) (41,120.0) (61,120.0) (81,120.0) (101,120.0) (121,120.0) (141,120.0) (161,120.0) (181,120.0) (201,120.0) (221,120.0) (241,120.0) (261,120.0) (281,120.0) (301,120.0) (321,120.0) (341,120.0) (361,120.0) (381,120.0) (401,120.0)};

        \addplot[thick, color=colConti, mark=*, mark size=1.2pt] coordinates {(1,0.2597604274749756) (21,0.2882338404655457) (41,0.35140663385391235) (61,0.41893357038497925) (81,0.4889250874519348) (101,0.5586663961410523) (121,0.6552823305130004) (141,0.7359154939651489) (161,0.8287474155426026) (181,0.9342723846435547) (201,1.0704388618469238) (221,1.171392560005188) (241,1.28218332529068) (261,1.3615889191627502) (281,1.47420175075531) (301,1.6003031969070434) (321,1.6882519483566285) (341,1.869273853302002) (361,2.009351670742035) (381,2.1343404173851015) (401,2.288579261302948)};

        \addplot[thick, color=colLola, mark=square*, mark size=1.2pt] coordinates {(1,0.0037623882293701173) (21,0.009159934520721436) (41,0.04863708019256592) (61,0.13488439321517945) (81,0.2836678147315979) (101,0.4857176184654236) (121,0.8112531781196595) (141,1.2583516597747804) (161,1.7793631434440613) (181,2.483343708515167) (201,3.6791560292243957) (221,4.659459483623505) (241,6.130179595947266) (261,7.428779661655426) (281,9.531965816020966) (301,11.60091712474823) (321,14.213600885868072) (341,17.066209626197814) (361,20.485576713085173) (381,24.664501202106475) (401,28.610709536075593)};

        \addplot[thick, color=colWoflan, mark=diamond*, mark size=1.2pt] coordinates {(1,0.008) (21,0.5043500000000001) (41,30.406200000000002) (61,41.92285) (81,58.54125) (101,66.00295) (121,43.991150000000005) (141,42.0291) (161,42.16715) (181,49.4807) (201,60.2767) (221,54.042300000000004) (241,57.35075) (261,67.0294) (281,55.035849999999996) (301,56.20935) (321,44.809050000000006) (341,67.001) (361,61.124449999999996) (381,60.0859) (401,72.06155)};

        \addplot[ultra thick, color=gray, opacity=0.25] coordinates {(1,120.0) (401,120.0)};

    }{0.45\textwidth}{xtick={1,101,...,401}, ytick={0,50,100,120}, height=0.38\textwidth}
    \hspace*{10pt}
    \plotcanvaswithextraargs{$N$}{net size}{
        \addplot[thick, black, mark=x, mark size=1.2pt] coordinates {(1,23.2) (21,81.5) (41,170.65) (61,254.5) (81,333.55) (101,398.1) (121,477.55) (141,547.05) (161,620.8) (181,692.85) (201,770.0) (221,856.45) (241,943.0) (261,1009.45) (281,1092.05) (301,1168.25) (321,1249.15) (341,1332.85) (361,1416.45) (381,1502.35) (401,1583.1)};

        \addplot[ultra thin, black!50!white, mark=x, mark size=1.2pt] coordinates {(1,13) (21,19) (41,95) (61,169) (81,229) (101,266) (121,355) (141,355) (161,458) (181,495) (201,581) (221,608) (241,720) (261,826) (281,850) (301,963) (321,1025) (341,1128) (361,1207) (381,1245) (401,1276)};

        \addplot[ultra thin, black!50!white, mark=x, mark size=1.2pt] coordinates {(1,52) (21,132) (41,271) (61,349) (81,465) (101,578) (121,638) (141,780) (161,784) (181,900) (201,1006) (221,1128) (241,1168) (261,1273) (281,1386) (301,1515) (321,1592) (341,1637) (361,1738) (381,1927) (401,1810)};
    }{0.45\textwidth}{xtick={1,101,...,401}, height=0.38\textwidth}
    \ourlegend

    \caption{Experiments on chained free-choice instances. The $x$-value denotes the number $N$
        of chained nets.
        Dark thick lines denote the mean, and light thin lines of the same color denote the minimum and maximum, respectively.
        For Woflan, the minimum line is slightly below the line of this work.
        For this work, the minimum and maximum lines are very close to the mean.
        \emph{Left:} The $y$-value denotes time for checking soundness of the 20 nets for each $N$.
        Marks on the gray line at $120s$ denote timeouts.
        \emph{Right:} The $y$-value denotes the size of generated nets.
    }\label{fig:fc-results}
\end{figure}

The results show that state-space exploration via LoLA is very fast
for moderate sizes, but does not scale as well. Continuous soundness
is in fact outperformed by LoLA for $N \leq 100$, but scales much
better, showing essentially linear growth in the given data range. For
instance, continuous soundness takes a mean of $0.25s$ for $N = 1$, a
mean of $1.07s$ for $N = 201$, and a mean of $2.28s$ for $N = 401$.

Woflan performs very well on the original instances, but times out
frequently for larger instances. Woflan checks so-called
$S$-coverability~\cite{verbeek2001diagnosing}. This is fast on many
instances, even large ones, but starts running into the
exponential-time worst case when instances get larger. For $N = 1$ and
$N = 21$, Woflan does not ever time out, while it times out for
roughly half of the instances in the range from $N = 201$ to $N =
    401$. Overall, we infer that for large free-choice workflow nets,
deciding soundness by checking continuous soundness can outperform
existing techniques, while the procedure is still competitive on
moderate instances.

\subsection{Synthetic instances}

In the previously discussed benchmark suite, nets are free-choice. So
structural and generalised soundness are equivalent
by \Cref{thm:fc:equiv}. We considered including a second suite of 590
non-free-choice Petri nets that represent processes of
the SAP reference model~\cite{mendling2006faulty}. However, all of
them turn out to be $1$-quasi-sound but not $1$-sound, so they
represent trivial cases for generalised and structural soundness:
simply checking $1$-soundness, or $1$-quasi-soundness and then
$1$-soundness, decides all instances. In order to have a wider variety
of challenging instances, we introduce several families of synthetic
workflow nets. The nets are simple to understand, but have large
numbers of reachable marking, so are challenging for
approaches relying on state-space exploration, \eg\ model checking.

\paragraph{Encoding arc weights.}

To simplify the presentation, we describe synthetic instances
utilizing arcs with weights. For
benchmarking, we removed the arc weights and instead input equivalent
weightless nets. To do so, we used an encoding that simulates
exponentially large weights by polynomially many transitions and
places (the encoding is explained in \Cref{ssec:arc:weights}).
It preserves (quasi-)soundness, but 
significantly increases the number of reachable markings. Indeed, our
synthetic instances are mostly trivial to solve by enumerating
reachable markings when arcs have weights, but become much harder to
decide when the encoding is used.\footnote{It is deliberately used to
make instances challenging, not to ensure compatibility with LoLA or
Woflan, as both support arc weights.} While much of the literature on
workflow nets does not consider nets with arc weights, implicit
structural encodings can occur in practice.

\subsubsection{Generalised soundness}

\paragraph{Benchmark instances.}

We introduce a synthetic family of nets where generalised soundness
appears to be challenging. The family $\{\pn_c\}_{c \in \Nn}$ is
defined at the top of \Cref{fig:synthetic}. Parameter $c \in \Nn$ is
the smallest value for which $\pn_c$ is $c$-unsound. From
$\imarked{c}$, the sequence $t_\initial^{c} t_r^{c+1}$ can be fired,
which leads to the deadlock $\{r \colon c+1\}$. Yet, when starting
with $k < c$ tokens in $\initial$, and firing $t_\initial^k$,
transitions $t_r$ and $t_\output$ can only be fired exactly $k$ times,
and $\fmarked{k}$ will be reached.

\begin{figure}
  \centering
  \begin{tabular}{c}
    \begin{tikzpicture}[node distance=0.5cm and 1.5cm,
                        scale=0.9, transform shape]
      \tikzstyle{cplace} = [place, minimum size=15pt]
      
      \node[cplace, label=left:{$\initial$}] (i) at (0.25,1) {};
      \node[transition, right=1.5cm of i, label=below:{$t_\initial$}] (t1) {};
      \path[->]
      (i) edge[]            node[] {}           (t1);

      \node[cplace, right=of t1, label=below:{$p$}] (p) {};
      \path[->]
      (t1) edge[]            node[above] {$c+1$}           (p);
      
      \node[transition, above=of p, label=left:{$t_r$}] (t3) {};
      \node[cplace, right=of t3, label=right:{$r$}] (r) {};
      
      \path[->]
      (p)  edge[]            node[left] {$c$}           (t3)
      (t3) edge[]            node[] {}           (r)
      ;
      
      \node[transition, right=1.5cm of p, label=below:{$t_\output$}] (tf) {};
      \node[cplace, right=1.5cm of tf, label=right:{$\output$}] (f) {};
      \path[->]
      (p)  edge[]            node[] {}           (tf)
      (r)  edge[]            node[] {}           (tf)
      (tf) edge[]            node[] {}           (f)
      ;    
    \end{tikzpicture} \\[-2pt]
    \midrule \\[-10pt]
    \begin{tikzpicture}[node distance=0cm and 1.25cm, transform shape, scale=0.9]
      \tikzstyle{cplace} = [place, minimum size=15pt]
        
      \node[cplace, label=left:{$\initial$}] (i) {};
      \node[transition, right=0.85cm of i, label=above:{$t$}]  (t) {};
      \node[cplace, right=0.85cm of t, label=right:{$\output$}] (f) {};
      
      \path[->]
      (i) edge[] node[above] {$c$} (t)
      (t) edge[] node[above] {$c$} (f)
      ;

      \node[cplace, label=left:{$\initial$}, right=1cm of f]   (i2)  {};
      \node[transition, right=0.85cm of i2, label=above:{$t$}]  (t2) {};
      \node[cplace, right=0.85cm of t2, label=right:{$\output$}] (f2) {};
      
      \path[->]
      (i2) edge[] node[above] {$c$}     (t2)
      (t2) edge[] node[above] {$c - 1$} (f2)
      ;

      \node[cplace, label=left:{$\initial$}, right=1cm of f2] (i3) {};
      \node[transition, right=0.5cm of i3, label=above:{$t_\initial$}] (ti) {};
      
      \node[cplace, above right=0cm and 0.5cm of ti, label={[xshift=-10pt, yshift=-5pt]above:{$u$}}] (u) {};
      \node[cplace, below right=0cm and 0.5cm of ti, label={[xshift=-10pt, yshift=5pt]below:{$d$}}] (d) {};
      
      \node[transition, right=0.5cm of u, label={[xshift=12pt, yshift=-7pt]above:{$t_u$}}] (tu) {};
      \node[transition, right=0.5cm of d, label={[xshift=12pt, yshift=7pt]below:{$t_d$}}] (td) {};
      
      \node[cplace, below right=0cm and 0.5cm of tu,
            label=right:{$\output$}] (f3) {};
      
      \path[->]
      (i3) edge[] node[above] {} (ti)
      (ti) edge[] node[above] {} (u)
      (ti) edge[] node[above] {} (d)
      ;
      
      \path[->]
      (u)  edge[] node[above] {$c$} (tu)
      (d)  edge[] node[above] {}    (tu)
      (tu) edge[] node[above] {}    (f3)
      ;
      
      \path[->]
      (d)  edge[bend right=15] node[below] {$2$} (td)
      (td) edge[bend right=15] node[above] {}    (d)
      (td) edge[]              node[above] {}    (f3)
      ;
    \end{tikzpicture}
  \end{tabular}
  \caption{\emph{Top}: A workflow net $\pn_c$ that is $c$-unsound and
    $k$-sound for all $k \in [1..c-1]$. \emph{Bottom}: Three families
    of instances. \emph{Bottom left:} $\pn_{\text{sound-}c}$ is
    quasi-sound and $\ell c$-sound for all $\ell \in
    \Nn$. \emph{Bottom center:} $\pn_{\neg \text{quasi-}c}$ is not
    structurally quasi-sound. \emph{Bottom right:} $\pn_{\neg
      \text{sound-}c}$ is $\ell c$-quasi-sound for all $\ell \in \Nn$,
    but not structurally sound.}\label{fig:synthetic}
\end{figure}
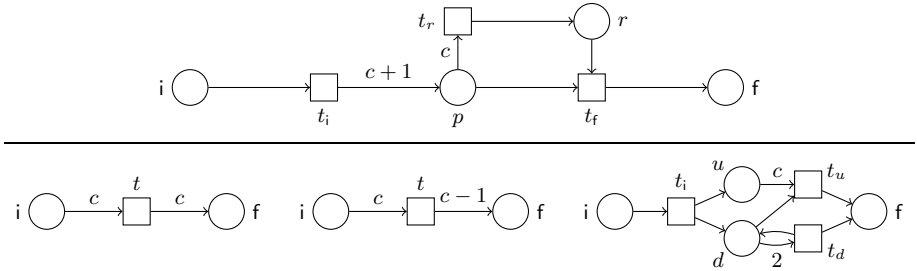

The naive approach to decide generalised soundness is to check
$k$-soundness for all $k$ until a counterexample is found or a bound
is exceeded. It is known that if a counterexample exists, then there
also is one of size at most exponential~\cite[Lemma~5.6
and~5.8]{BMO22}. The approach we chose for semi-deciding generalised
soundness is to check continuous soundness. Recall that continuous
soundness is a necessary (albeit not sufficient) condition, as shown
in~\Cref{thm:cont-necessary}.

In our evaluation, we used Woflan and LoLA to check generalised
soundness of the family for different $c$ by checking
$1$-sound, \dots, $c$-soundness, and compared the result to the time
needed for testing continuous soundness. Our main goal is to evaluate
whether checking continuous soundness is efficient enough to serve as
an inexpensive way to witness generalised unsoundness for nontrivial
instances.

\paragraph{Results.}

\begin{figure}
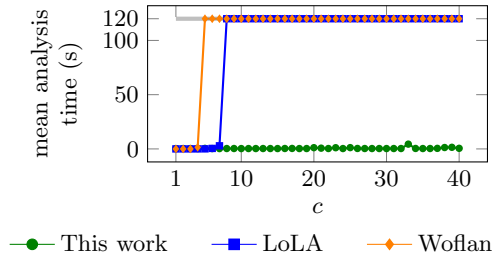

    \centering
    \plotcanvaswithextraargs{$c$}{\begin{tabular}{c} mean analysis \\ time (s)\end{tabular}}{
        \addplot[thick, color=colConti, mark=*, mark size=1pt] coordinates {(1,0.2959151268005371) (2,0.30132246017456055) (3,0.31386709213256836) (4,0.2971646785736084) (5,0.5014362335205078) (6,0.30652689933776855) (7,0.3016331195831299) (8,0.36017775535583496) (9,0.37421655654907227) (10,0.3196728229522705) (11,0.3366086483001709) (12,0.3380444049835205) (13,0.3582794666290283) (14,0.3457765579223633) (15,0.32504701614379883) (16,0.3910665512084961) (17,0.43927931785583496) (18,0.33383798599243164) (19,0.34503602981567383) (20,1.143829584121704) (21,0.6930875778198242) (22,0.3663506507873535) (23,1.1087050437927246) (24,0.3396890163421631) (25,1.1844286918640137) (26,0.41590428352355957) (27,0.3689877986907959) (28,0.34794139862060547) (29,0.3582954406738281) (30,0.3368706703186035) (31,0.360119104385376) (32,0.3632214069366455) (33,4.406553268432617) (34,0.42772960662841797) (35,0.3476724624633789) (36,0.3776247501373291) (37,0.5114471912384033) (38,1.285477638244629) (39,1.4889757633209229) (40,0.5513830184936523)};

        \addplot[thick, color=colLola, mark=square*, mark size=1pt] coordinates {(1,0.003349781036376953) (2,0.006882429122924805) (3,0.011688709259033203) (4,0.019080162048339844) (5,0.04892849922180176) (6,0.46179842948913574) (7,2.867579221725464) (8,120.0) (9,120.0) (10,120.0) (11,120.0) (12,120.0) (13,120.0) (14,120.0) (15,120.0) (16,120.0) (17,120.0) (18,120.0) (19,120.0) (20,120.0) (21,120.0) (22,120.0) (23,120.0) (24,120.0) (25,120.0) (26,120.0) (27,120.0) (28,120.0) (29,120.0) (30,120.0) (31,120.0) (32,120.0) (33,120.0) (34,120.0) (35,120.0) (36,120.0) (37,120.0) (38,120.0) (39,120.0) (40,120.0)};

        \addplot[thick, color=colWoflan, mark=diamond*, mark size=1pt] coordinates {(1,0.002) (2,0.005) (3,0.013) (4,1.502) (5,120.0) (6,120.0) (7,120.0) (8,120.0) (9,120.0) (10,120.0) (11,120.0) (12,120.0) (13,120.0) (14,120.0) (15,120.0) (16,120.0) (17,120.0) (18,120.0) (19,120.0) (20,120.0) (21,120.0) (22,120.0) (23,120.0) (24,120.0) (25,120.0) (26,120.0) (27,120.0) (28,120.0) (29,120.0) (30,120.0) (31,120.0) (32,120.0) (33,120.0) (34,120.0) (35,120.0) (36,120.0) (37,120.0) (38,120.0) (39,120.0) (40,120.0)};

        \addplot[ultra thick, color=gray, opacity=0.5] coordinates {(1,120.000) (40,120.000)};
    }{0.5\textwidth}{xtick={1,10,20,30,40}, ytick={0,50,100,120}, height=0.3\textwidth, legend columns=3, legend style={at={(0.5,-0.5)},anchor=north}}

    \ourlegend
    \caption{Time to check generalised soundness of $\pn_c$ for different values of $c$.
    Marks on the gray line at $120s$ denote timeouts.}\label{fig:synthetic-gen-sound-results}
\end{figure}

\Cref{fig:synthetic-gen-sound-results} depicts the results. Woflan and
LoLA show good performance for small values of $c$, but do not scale
well to larger values. They respectively time out for $c \geq 5$ and
$c \geq 8$. The instances are not free-choice, so LoLA and Woflan need
to explore the state-space for each $k \leq c$, which becomes
infeasible. For $c \geq 14,$ Woflan cannot even check $1$-soundness
within the time limit. LoLA can check $1$- and $2$-soundness for $c
\leq 28$, but cannot handle $2$-soundness for larger $c$. Continuous
soundness is efficiently verifiable even for $c = 40$. In particular,
we need less than $5s$ on all instances. The greatest time is at $c =
33$.  Further, at most $1s$ is needed on 34 out of 40 instances (mean
of $0.6s$).

\subsubsection{Structural soundness}

\paragraph{Benchmark instances.}

For structural soundness, recall that our decision procedure is based
on checking structural quasi-soundness
and obtaining some lower bound for the smallest number
for which the net is quasi-sound. Thus, we want to test on both
benchmark instances that are structurally quasi-sound and those that
are not.
We introduce three families of non-free-choice nets
for which structural soundness appears challenging.
These instances are defined at the bottom of
\Cref{fig:synthetic}. We respectively denote them
$\pn_{\text{sound-}c}$ (left), $\pn_{\neg \text{quasi-}c}$ (center)
and $\pn_{\neg \text{sound-}c}$ (right). We claim that:
$\pn_{\text{sound-}c}$ is $\ell c$-sound for all $\ell \in
    \Nn$; $\pn_{\neg \text{quasi-}c}$ is not structurally quasi-sound;
$\pn_{\neg \text{sound-}c}$ is $\ell c$-quasi-sound for all $\ell \in
    \Nn$, not $k$-quasi-sound for any other number $k \in \Nn$, and
not structurally sound.

For the experiments, our goal is twofold.
First, we want to evaluate whether utilizing continuous reachability to
decide structural quasi-soundness is more efficient than using
the known reduction to reachability described in~\cite[Lemma~2.1]{ctiplea2005structural}.
Woflan does not directly support checking reachability,
so we only compare with LoLA. Second, we want to evaluate whether the lower bound for the smallest
number for which the net is quasi-sound, which we dubbed $k_{\pn,\Qpos}$ towards the end of \Cref{sec:struct:sound},
is close to the actual smallest number, dubbed $k_{\pn}$.

A caveat of this evaluation is that we evaluate only on our synthetic instances,
and that computing $k_{\pn,\Qpos}$ is only one step in deciding structural soundness.
However,
we think that the evaluation on these hard synthetic instances can give
insights into the applicability on nontrivial real-world instances.

\paragraph{Results.}

\Cref{fig:synth-struct-results} compares the time needed to verify
structural reachability for LoLA and our prototype. For small
instances, LoLA sometimes performs very well, but we scale better for large
values. Of particular note is that in the absence of quasi-soundness,
LoLA will generate an infinite state-space, so will generally run out of time or memory.
In particular, LoLA times out for all $c$ on
$\N_{\neg \text{quasi-}c}$. It also times out for $c \geq 32$ on $\N_{\neg
  \text{sound-}c}$. On the other hand, continuous soundness never times out for the given
values of $c$. In fact, when we tested continuous soundness for much
larger values of $c$, we found that our implementation of continuous
reachability decides structural quasi-soundness for $N_{\neg
  \text{quasi-}c}$ in under $2s$ for $c = 20~000~000$.

We further found that for all instances, $k_{\pn,\Qpos} = k_\pn$, that is, our  lower bound exactly matches the 
smallest number for which the net is quasi-sound.
Thus, it only remains to decide $k_{\pn,\Qpos}$-quasi-soundness and $k_{\pn,\Qpos}$-soundness in order to decide structural soundness.
This is in contrast to the naive approach, which starts at $k=1$ and
checks $k$-quasi-soundness for each value up to $k_{\pn,\Qpos}$.

\begin{figure}
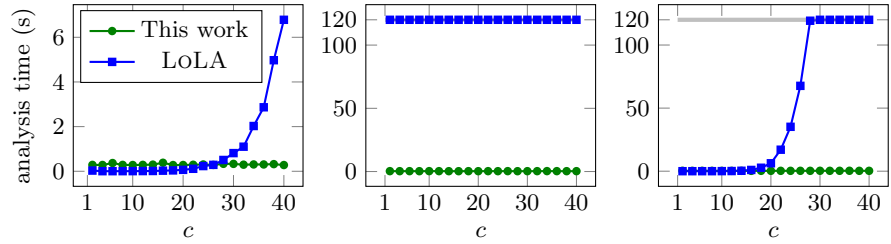

    \centering
    \input{figures/data/synthetic-struct-reach_sound.tex}%
    \input{figures/data/synthetic-struct-reach_unreach.tex}%
    \input{figures/data/synthetic-struct-reach_unsound.tex}%
    
    \caption{Time taken vs parameter $c$ for checking structural quasi-soundness using the reduction to reachability, and utilizing our approach to compute $k_{\pn,\Qpos}$, for each of the three families at the bottom of \Cref{fig:synthetic}: $\pn_{\text{sound-}c}$ (\emph{left}),$\pn_{\neg \text{quasi-}c}$ (\emph{center}), $\pn_{\neg \text{sound-}c}$ (\emph{right}).
        Note that the axis ranges differ.
        Marks on the gray line at $120s$ denote timeouts.}\label{fig:synth-struct-results}
\end{figure}

\section{Conclusion}
\label{sec:conclusion}
In this work, we have shown how reachability relaxations allow to
efficiently semi-decide generalised and structural soundness. Our
approach combines nicely with reduction rules, as they all preserve
relaxations. In particular, we have introduced continuous soundness as
an approximation of generalised soundness, and shown that it coincides
with other types of soundness for free-choice nets.

As part of future work, we plan to migrate our prototype into the
process mining framework ProM, to make the algorithms available to
practitioners.

\subsection*{Acknowledgements}
We thank Dirk Fahland and Eric Verbeek for their help with Woflan.
Michael Blondin was supported by a Discovery Grant from the Natural Sciences and Engineering Research Council of Canada (NSERC), and by the Fonds de recherche du Québec – Nature et technologies (FRQNT).

\bibliographystyle{splncs04}
\bibliography{references}

\clearpage
\appendix
\label{sec:appendix}
\section{Appendix}
\subsection{Missing proofs of \Cref{sec:relaxations}}

\lemContEquivReach*

\begin{proof}
  $\Leftarrow$) Let $b \cdot m \trans{\pi} b \cdot \m'$. Let $\beta
    \defeq 1 / b$. Let us prove that $\m \trans{\beta \cdot \pi} \m'$.
  To do so, we show that $b \cdot \m \trans{\FromTo{\pi}{1}{n}}
    \m_{n}$ implies $\m \ctrans{\FromTo{\pi}{1}{n}} \beta \cdot
    \m_{n}$. Let us proceed by induction on $n$. Assume that \[b \cdot m
    \trans{\FromTo{\pi}{1}{n}} \m_{n} \trans{t_{n+1}} \m_{n+1} \text{
      where } t_{n+1} \defeq \pi[n+1].\] By induction hypothesis, we
  have $\m \ctrans{\FromTo{\pi}{1}{n}} \beta \cdot \m_{n}$. Note that
  $\m_{n+1} = \m_{n} + \effect{t_{n+1}}$. So, $\beta \cdot \m_{n} +
    \beta \cdot \effect{t_{n+1}} = \beta \cdot \m_{n+1}$. Thus, $\beta
    \cdot t_{n+1}$ has the right effect to lead from $\m_n$ to
  $\m_{n+1}$. It only remains to show that $\beta \cdot t_{n+1}$ is
  enabled at $\m_{n}$. Note that $t_{n+1}$ is enabled at $\m_n$, hence
  by definition, $\m_{n}[p] \geq \pre{t_{n+1}}[p]$ for all $p \in
    P$. It follows that $\beta \cdot \m_{n}[p] > \beta \cdot
    \pre{t_{n+1}}[p]$, so $\beta \cdot t_{n+1}$ is enabled in $\m_n$.

  $\Rightarrow$) Let $\m \ctrans{\pi} \m'$. Let $\beta$ be the product
  of the scaling factors denominators along $\pi$. Let us show that $b
    \cdot \m \trans{\pi'} b \cdot \m'$. We establish the following for
  all $n$:
  \[
    \text{ if } \m \ctrans{\FromTo{\pi}{1}{n}} \m_n,
    \text{ then there exists } \pi'_n
    \text{ such that } b \cdot \m \trans{\pi'_{n}} b \cdot \m_n.
  \]
  Assume this holds for some $n$. Let $\alpha \cdot t_{n+1} = \pi[n]$.
  We show the following:
  \[
    \text{if } \m_{n} \ctrans{\alpha t_{n+1}} \m_{n+1},
    \text{then } b \cdot \m_{n} \trans{(t_{n+1})^{b \cdot \alpha}} b \cdot
    \m_{n+1}.
  \]
  First, let us argue that $b \cdot \alpha$ is an integer. Note that
  by the fact that scaling factors are chosen from $\ZeroOne$, it
  follows that $\alpha$ can be written as $u / d$ for some $u, d \in
    \N$ where $d \neq 0$. Further, note that $b$ was chosen as the
  product of all denominators of scaling factors along $\pi$. In
  particular, $d$ is a factor of $b$, so we have $b = d \cdot b'$ for
  some $b' \in \N$, and thus $b \cdot \alpha = d \cdot b' \cdot u / d
    = b' \cdot u$. Next, let us argue that $(t_{n+1})^{b \cdot \alpha}$
  has the right effect to lead from $b \cdot \m_{n}$ to $b \cdot
    \m_{n+1}$. Note that $\m_{n+1} = \m_{n} + \alpha \cdot
    \effect{t_{n+1}}$. So, $b \cdot \m_{n+1} = b \cdot \m_{n} + b \cdot
    \alpha \effect{t_{n+1}} = b \cdot \m_{n} + \effect{(t_{n+1})^{b
          \cdot \alpha}}$. It remains to argue that $(t_{n+1})^{b \cdot
        \alpha}$ is fireable from $b \cdot \m_{n}$. By $\m_{n}
    \ctrans{\alpha t_{n+1}} \m_{n+1}$, it follows that $\m_{n}[p] \geq
    \alpha \pre{t_{n+1}}[p]$ for all $p \in P$. Since $b \in \N$, it is
  the case that $b \cdot \m_{n}[p] \geq b \cdot \alpha
    \pre{t_{n+1}}[p]$, and hence we are done. \qed
\end{proof}

\propReductions*

\begin{proof}
  We will informally present the rules by the properties they
  preserve. For a formal definition of the rules, we refer
  to~\cite[Sect.~4.2]{bride2017reduction}. Most arguments apply to all
  $\D \in \{\N, \Z, \Qpos\}$ in the same way, thus usually we will not
  make a cumbersome case analysis.

  \medskip\noindent\emph{Rule $R_1$ (place removal).} This rule
  removes a place $p \in P$. Thus, $P' = P \setminus \set{p}$. It is
  guaranteed that there exist places $\{g_1, \ldots, g_n\} \subseteq
    P'$ such that the number of tokens in $p$ is the sum of tokens in
  those places. Hence, it suffices to define $R' \defeq \set{g_1,
      \ldots, g_n}$.

  \medskip\noindent\emph{Rules $R_2$ (transition removal) and $R_3$
    (loop removal).} For these rules, no place is removed and the
  reachability relation is preserved.

  \medskip\noindent\emph{Rules $R_4$ (transition-place removal) and
    $R_5$ (place-transition removal)}. These rules remove a place $p$
  and its only input (for $R_4$) or output (for $R_5$) transition
  $t$. Transition $t$ is merged with the output (for $R_4$) or input
  (for $R_5$) transitions. Thus, intuitively, the new transitions in
  $\pn'$ immediately consume a token whenever it was put in $p$. This
  clearly proves that $\m' \dreach \n'$ in $\pn'$ iff $\pi_0(\m')
    \dreach \pi_0(\n')$ in $\pn$. Moreover, the requirements on when the rule can be applied
  imply either $\pre{t}[p] = 1 \implies \pre{t}[p'] = 0$; or $\post{t}[p] = 1 \implies \post{t}[p'] =0$.

  It remains to prove the final part when $\D \neq \Z$. Suppose there exists $\m$ such that
  $\imarked{1} \dreach \m \not \dreach \fmarked{1}$ in $\pn$. Suppose
  first, that there exists $\n$ such that $\m \dreach \n$ and $\n[p] =
    0$. By the previous case, we have $\pi(\n) \not \dreach
    \fmarked{1}$, as otherwise we reach the contradiction $\m \dreach \n
    \dreach \fmarked{1}$. We define $\m' \defeq \pi(\n)$. In the second
  case, we can assume that for all $\n$, $\m \dreach \n$ implies
  $\n[p] > 0$ (here we use $\D \neq \Z$). In particular, $\m[p] >
    0$. Let $T_p \defeq \set{t \in T \mid \post{t}[p] = 1}$. We conclude
  from the additional constraints on $R_4$ and $R_5$
  (see~\cite{bride2017reduction}). These imply that in our case for
  every $t \in T_p$ and for all $r \in P$:
  \begin{enumerate}
    \item if $r \neq p$, then $\post{t}[r] = 0$;\label{item1}

    \item if $\pre{t}[r] = 1$ and $t' \neq t$, then $\pre{t'}[r] =
            0$.\label{item2}
  \end{enumerate}
  Let $\rho$ be the run witnessing $\imarked{1} \dreach \m$. Let
  $\rho'$ be the subrun of transitions in $T_p$ that occur in $\rho$
  (it is nonempty since $\m[p] > 0$). By \Cref{item1} we can remove
  (or downscale if $\D = \Qpos$) a suffix of $\m[p]$ transitions in
  $\rho'$ (because it removes tokens only from $p$). We obtain a
  marking $\m_1$ such that: $\m_1[p] = 0$; the tokens in $\m_1[r]$ for
  all removed $t \in T_p$ such that $\pre{t}[r] = 1$ have increased
  accordingly; and $\m_1[p'] = \m[p]$ otherwise. We claim that $\m' =
    \pi(\m_1)$ satisfies the proposition. Indeed, if there is a run
  $\pi(\m_1) \dreach \fmarked{1}$ in $\pn'$ then by \Cref{item2} we
  can extract from this a run $\m \dreach \fmarked{1}$ in $\pn$, which
  would be a contradiction.

  \medskip\noindent\emph{$R_6$ (ring removal).} This rule merges a set
  of places $\{p_1, \ldots, p_k\} \subseteq P$ into a single place
  $p_1$\footnote{In~\cite{bride2017reduction}, $p_1$ is also removed
    and a new place $p$ is added, but this is trivially
    equivalent.}. Thus, $P' = P \setminus \set{p_2, \ldots, p_k}$. The
  conditions are that the tokens can be transferred arbitrarily
  between the places $p_1, \ldots, p_k$, which is enough to prove the
  proposition. \qed
\end{proof}

\subsection{Missing proofs of \Cref{sec:gen:sound}}

\propIntBoundRed*

\begin{proof}
  We will need to invoke \Cref{prop:zunbounded:charac} which is stated
  after \Cref{prop:int:bound:red} in the main text. Note that this ordering
  is simply for the sake of presentation, there is no circular
  dependency, the proof of \Cref{prop:zunbounded:charac} is
  self-contained.

  By \Cref{prop:zunbounded:charac}, being integer unbounded is
  equivalent to the existence of $\bv > \vec{0}$ such that $\0 \zreach
    \bv$. Let $\pn$ and $\pn'$ be the workflow nets before and after the
  reduction. We invoke \Cref{prop:reductions} depending on the applied
  reduction rule, and show that $\pn$ is integer unbounded iff $\pn'$
  is integer unbounded.

  \begin{itemize}
    \item \emph{Rule $R_1$}. Suppose $\0 \zreach \bv > \vec{0}$ in
          $\pn$. We have $\pi(\bv) > \0$, since $\bv[p] = \sum_{r \in R'}
            \bv[r]$. Thus, $\pi(\bv)[r] > \0$ for at least one $r \in R'$. The
          converse implication is trivial.

    \item \emph{Rules $R_2$ and $R_3$}. This is trivial because
          $\zreach$ is preserved.

    \item \emph{Rules $R_4$ and $R_5$}. We have $\m' \zreach \n'$ in
          $\pn'$ iff $\pi_0(\m') \zreach \pi_0(\n')$ in $\pn$. Thus, if
          $\vec{0} \zreach \bv' > \vec{0}$ in $\pn'$, then $\vec{0} \zreach
            \pi_0(\bv') > \vec{0}$ in $\pn$. Conversely, suppose that $\vec{0}
            \ztrans{\rho} \bv > \0$ in $\pn$. If $\bv[p] = 0$, then we are
          done. Otherwise, by \cref{prop:reductions} for all $t \in T$ and
          $p' \in P'$: either $\pre{t}[p] = 1 \implies \pre{t}[p'] = 0$; or
          $\post{t}[p] = 1 \implies \post{t}[p'] =0$. Let us assume the
          former and let $T_p \defeq \set{t \in T \mid \pre{t}[p] = 1}$. By
          removing $\bv[p]$ transitions from $T_p$ in $\rho$, we get
          $\vec{0} \zreach \bv' > \0$ and $\bv'[p] = 0$. Thus, $\0 \zreach
            \pi(\bv') > \0$ in $\pn'$ as required. In the latter case, we
          proceed similarly, but one need to add some transitions to $\rho$
          that will move the tokens from $p$ to other places.

    \item \emph{Rule $R_6$.} In this case, if $\m[p_i] = \n[p_i] = 0$
          for $i > 1$, then $\m \zreach \n$ in $\pn$ iff $\pi(\m) \zreach
            \pi(\n)$ in $\pn'$. Thus, $\0 \zreach \bv' > \0$ in $\pn'$ clearly
          implies $\0 \zreach \bv > \0$ in $\pn$. Conversely, if $\0 \zreach
            \bv > \0$ in $\pn$, then we know that $\bv \zreach \bv_1$ where
          $\bv_1[p_1] = \sum_{i = 1}^k \bv_1[p_i]$ and $\bv_1[p_i] = 0$ for
          $i > 1$. So, $\0 \zreach \pi(\bv_1) > \0$ in $\pn'$. \qed
  \end{itemize}
\end{proof}

\propIntBoundLin*

\begin{proof}
  Let $\pn = (P, F, T)$ be a Petri net. By
  \Cref{prop:zunbounded:charac}, $(\pn, \m)$ is integer bounded iff
  there exists $\m' > \vec{0}$ such that $\vec{0} \zreach \m'$. The
  latter amounts to the existence of $\pi \in T^*$ such that
  $\effect{\pi} > \vec{0}$. So, this is equivalent to this system:
  $\exists \vec{x} \in \N^T : \sum_{t \in T} \vec{x}[t] \cdot
    \effect{t} > \vec{0}$. It is readily seen that this system is
  equivalent to the one where $\vec{x} \in \Qpos^T$. Indeed, by
  homogeneity ($\vec{0}$ on the right-hand side), a rational solution
  can be scaled so that it becomes an integral solution.

  The polynomial time decidability of integer boundedness follows
  immediately from the fact that linear programming can be solved in
  polynomial time (\eg, see~\cite{Sch86}). \qed
\end{proof}

\propContBoundRed*

\begin{proof}
  Let $\pn$ and $\pn'$ be the workflow nets before and after the
  reduction. We invoke \cref{prop:reductions} depending on the applied
  reduction rule and show that $\pn$ is continuous sound iff $\pn'$ is
  continuous sound.

  \begin{itemize}
    \item \emph{Rule $R_1$}. Suppose $\imarked{1} \creach \m' \not
            \creach \fmarked{1}$ in $\pn'$. Let $\m$ be such that $\pi(\m) =
            \m'$ and $\m[p] = \sum_{r \in R'} \m[r']$. Then $\imarked{1}
            \creach \m$ and $\m \creach \fmarked{1}$ would imply $\m'
            \creach \fmarked{1}$, which is a contradiction. The converse
          implication is trivial.

    \item \emph{Rules $R_2$ and $R_3$}. This is trivial because
          $\creach$ is preserved.

    \item \emph{Rules $R_4$ and $R_5$}. Suppose $\imarked{1} \creach
            \m' \not \creach \fmarked{1}$ in $\pn'$. We have $\imarked{1}
            \creach \pi_0(\m')$. If $\pi_0(\m') \creach \fmarked{1}$ in
          $\pn$ then, since $\fmarked{1}[p] = 0$, we obtain $\m' \creach
            \fmarked{1}$ in $\pn'$, which is a contradiction. The other
          implication is explicitly written in \Cref{prop:reductions}.

    \item \emph{Rule $R_6$}. Suppose $\imarked{1} \creach \m \not
            \creach \fmarked{1}$ in $\pn$. We have $\m \zreach \m_1$ where
          $\m_1[p_1] = \sum_{i = 1}^k \bv_1[p_i]$ and $\m_1[p_i] = 0$ for
          $i > 1$. If $\pi(\m_1) \creach \fmarked{1}$, then $\m_1 \creach
            \fmarked{1}$, which is a contradiction. The other implication is
          trivial. \qed
  \end{itemize}
\end{proof}

\thmCoNP*

\begin{proof}[of coNP-hardness]
  Recall that, in the main text, we have defined a workflow net
  $\pn_\varphi$ from a formula in DNF, and claimed that
  $\pn_{\varphi}$ is continuously sound iff $\varphi$ is a
  tautology. It remains to show the implication from right to left.

  $\Rightarrow$) Suppose $\varphi$ is a tautology. Let us first make
  an observation. Consider some sequence $b_1, \ldots, b_m \in \{0,
    1\}$, and the marking $\m = \set{p_{i,b_i} \colon 1 \mid i \in
      [1..m]}$. Since $\varphi$ is a tautology, there exists a clause
  $C_j$ that satisfies the assignment $x_i \defeq b_i$. Let $i_1,
    \ldots, i_{\ell}$ be the indices of variables not occurring in
  $C_j$. It is easy to see that
  \begin{align*}
    \imarked{i}
    \trans{t_\text{init} v_{1, b_1} \cdots v_{m, b_m}}
    \m
    \trans{c_j \overline{v}_{i_1,b_{i_1}} \cdots
    \overline{v}_{i_{\ell}, b_{i_{\ell}}}}
    \set{r_i \colon 1 \mid i \in [1..m]}
    \trans{t_{\text{fin}}}
    \fmarked{1}.
  \end{align*}
  By \cite[Lemma~12(1)]{FH15}, we rescale the continuous run, \ie\ for
  all $\alpha \in \ZeroOne$:
  \begin{align}
    \imarked{\alpha} \creach \alpha \m \creach \sum_{i=1}^m\set{r_i
      \colon \alpha} \creach \fmarked{\alpha}.\label{eq:1}
  \end{align}

  Let us establish some invariants. Let $A_i \defeq \set{\initial,
      p_{i,?}, p_{i,1}, p_{i,0}, r_i, \output}$ and $B_i \defeq
    \set{\initial, p_{\text{cl}}, q_i, r_i, \output}$. First, for all
  transition $t \in T$ and all index $i \in [1..m]$, we have
  \begin{align*}
    \sum_{p \in A_i} \pre{t}[p] = \sum_{A_i} \post{t}[p],
    \text{ and }
    \sum_{p \in B_i} \pre{t}[p] = \sum_{p \in B_i} \post{t}[p].
  \end{align*}
  We say that a marking $\n$ is \emph{reachable} if $\imarked{1}
    \creach \n$. From the above invariants, it follows that every
  reachable marking $\n$ satisfies
  \begin{align}
    \sum_{p \in A_i}\n[p] = 1
    \text{ and }
    \sum_{p \in B_i}\n[p] = 1.\label{eq:2}
  \end{align}

  Note that, from \Cref{eq:2}, every reachable marking $\n$ satisfies
  \begin{align}
    \n[p_{i,?}] + \n[p_{i,1}] + \n[p_{i,0}]
    = \n[p_{\text{cl}}] + \n[q_{i}].\label{eq:3}
  \end{align}

  We further have this remaining invariant for all $t \in T$ and $i, j
    \in [1..m]$:
  \begin{align*}
    \pre{t}[q_i] + \pre{t}[r_i] + \post{t}[q_i] + \post{t}[r_i] =
    \pre{t}[q_j] + \pre{t}[r_j] + \post{t}[q_j] + \post{t}[r_j].
  \end{align*}
  Since all places $q_i$ and $r_i$ are empty in $\imarked{1}$, every
  reachable marking $\n$ satisfies:
  \begin{align}
    \n[q_i] + \n[r_i] = \n[q_j] + \n[r_j].\label{eq:4}
  \end{align}

  We are ready to prove continuous soundness. Let $\n$ be a reachable
  marking. By \Cref{eq:1}, we can assume w.l.o.g.\ that $\n[\initial]
    = 0$, as we can move $\alpha$ remaining token to
  $\output$. Similarly, we can assume w.l.o.g.\ that $\n[p_{i,?}] = 0$
  for all $i \in [1..m]$ as otherwise we can fire transition $v_{i,1}$
  or $v_{i,0}$ properly scaled (the choice is
  irrelevant). Consequently, by \Cref{eq:3}, we have $\n[p_{i,1}] +
    \n[p_{i,0}] \ge \n[q_{i}]$. Therefore, by firing transitions
  $\overline{v}_{i,0}$ and $\overline{v}_{i,1}$, scaled appropriately,
  we obtain $\n \creach \n'$ with $\n'[q_i] = 0$ for all $i \in
    [1..m]$. By \Cref{eq:4}, $\n'[r_i] = \n'[r_j]$ for all $i,j \in
    [1..m]$. Hence, by firing $t_{\text{fin}}$ scaled by $\n'[r_1]$, we
  get $\n' \creach \n''$ where $\n''$ has zero token in each place,
  except possibly places $P_\text{var}' \defeq \{p_{i,b} \mid i \in
    [1..m], b \in \{0, 1\}\}$, place $p_\text{cl}$ and place $\output$.

  Let us explain how to empty $P_\text{var}' \cup \{p_\text{cl}\}$, if
  this is not already the case. For each $i \in [1..m]$, among places
  $p_{i,1}$ and $p_{i,0}$, we write $p_{i,\max}$ and $p_{i,\min}$ so
  that $\n''[p_{i,\max}] \ge \n''[p_{i,\min}]$ (if they are equal,
  then the choice is not important). Let $S \defeq \set{p_{i,\min}
    \mid i \in [1..m], \n''[p_{i,\min}] > 0}$. We consider two cases.

  \emph{Case 1: $S = \emptyset$.} By the left part of \Cref{eq:2}, and
  by \Cref{eq:3}, the following holds for all $i, j \in [1..m]$:
  \begin{align}
    \n''[p_{i,1}] + \n''[p_{i,0}] = \n''[p_{j,1}] + \n''[p_{j,0}] =
    \n''[p_\text{cl}].\label{eq:5}
  \end{align}
  Thus, there exist $\alpha \in \ZeroOne$ and $b_1, \ldots, b_m \in
    \{0, 1\}$ such that $\n''[p_\text{cl}] = \n''[p_{i, b_i}] = \alpha$
  and $\n''[p_{i, \neg b_i}] = 0$ for all $i \in [1..m]$. Since
  $\varphi$ is a tautology, we can fire some transition $c_j$ scaled
  by $\alpha$, which empties place $p_\text{cl}$, and consequently
  $P_\text{var}'$ as well by \Cref{eq:5}.

  \emph{Case 2: $S \neq \emptyset$.} Let $i \in [1..m]$ be such that
  $\n''[p_{i, \min}] > 0$ is minimal, and let $\alpha \defeq
    \n''[p_{i, \min}]$. Let $\n''' \defeq \set{p_{\text{cl}} \colon \alpha, p_{i, \min} \colon
      \alpha} + \set{p_{j, \max} \colon \alpha \mid j \neq i}$. Note
  that $\n''' \le \n''$. We can apply \Cref{eq:1} and obtain
  \[
    \n'' = (\n'' - \n''') + \n''' \creach (\n'' - \n''') + \set{\output
      \colon \alpha}.
  \]
  Performing this operation decreases the size of $S$. Hence, it can
  be repeated at most $m$ times until $S$ becomes empty, which has
  been handled in case~1. \qed
\end{proof}

\subsection{Missing proofs of \Cref{sec:struct:sound}}

\propIntegerK*

\begin{proof}
  Let $\pn = (P, T, F)$ be a workflow net. Let us first establish
  $k_{\pn,\Z} \leq k_{\pn,\Qpos}$. Let $\pi = \lambda_1 t_1 \cdots
    \lambda_n t_n$ be a continuous run such that $\imarked{k}
    \ctrans{\pi} \omarked{k}$ and $\vec{\pi} \in \N^T$. In particular,
  we have
  \begin{align*}
    \omarked{k}
     & = \imarked{k} + \sum_{i \in [1..n]} \lambda_i \cdot \effect{t_i} \\
     & = \imarked{k} + \sum_{t \in T} \sum_{i \in [1..n] : t_i = t}
    \lambda_i \cdot \effect{t}                                          \\
     & = \imarked{k} + \sum_{t \in T} \vec{\pi}[t] \cdot \effect{t}.
  \end{align*}
  As $\vec{\pi} \in \N^T$, we obtain $\imarked{k} \ztrans{\pi}
    \omarked{k}$. Consequently, $k_{\pn,\Z} \leq k_{\pn,\Qpos}$.

  The inequality $k_{\pn,\Qpos} \leq k_\pn$ follows immediately from
  the fact that $\imarked{k} \trans{\pi} \omarked{k}$ implies
  $\imarked{k} \ctrans{\pi} \omarked{k}$ (with all scaling factors set to $1$).

  It remains to argue that $k_{\pn,\Z}$ and $k_{\pn,\Qpos}$ can be
  obtained as described. By definition of integer reachability,
  $k_{\pn,\Z}$ is the value obtained from this program:
  \[
    \min k \text{ subject to } k \in \Nn, \vec{x} \in \N^T
    \text{ and }
    \imarked{k} + \sum_{t \in T} \vec{x}[t] \cdot \effect{t} =
    \omarked{k}.
  \]
  For $k_{\pn,\Qpos}$, we use the fact that there is polynomial-time
  constructible formula $\psi_{\pn}$ from existential linear
  arithmetic such that $\psi(\m, \m', \vec{x})$ holds iff there is
  a continuous run $\pi$ that satisfies $\m \ctrans{\pi} \m'$ and
  $\vec{x} = \vec{\pi}$~\cite{BFHH17}. So, it suffices to take
  \[
    \varphi(k) \defeq \exists \vec{x} \in \N^T : \psi(\imarked{k},
    \fmarked{k}, \vec{x}) \tag*{\qed}.
  \]
\end{proof}

\subsection{Missing proofs of \Cref{sec:freechoice}}

Recall the following unproven lemma from the main text.

\lemAllThree*

For the sake of readability, we prove each item of
\Cref{lem:all:three} as its own lemma.

\begin{lemma}
  Let $\pn = (P, T, F)$ be a free-choice Petri net, and let $\m \in
    \N^P$. It is the case that $\m \reach \m'$ for some marking $\m'$
  such that $L(\m') = F(\m')$.
\end{lemma}

\begin{proof}
  If $F(\m) = L(\m)$ holds, then we are done by taking $\m' \defeq
    \m$. Otherwise, let $t \in F(\m) \setminus L(\m)$. Since $t$ is not
  live in $(\pn, \m)$, there exists a marking $\m' \in \N^P$ that
  satisfies $\m \reach \m'$ and $t \notin F(\m')$. Therefore, we have
  $F(\m') \subseteq F(\m) \setminus \{t\}$. This means that $|F(\m')|
    < |F(\m)$. Since $L(\m') \subseteq F(\m')$, we can repeat this
  argument (up to $|T|$ times) until obtaining $L(\m') = F(\m')$. \qed
\end{proof}

For a run $\sigma$, let us define $\parikh{\sigma}\colon T \to \N$,
where for each $t \in T$, $\parikh{\sigma}[t]$ is the number of times
$t$ occurs in $\sigma.$

\begin{lemma}
  Let $\pn$ be a free-choice workflow net, let $c \in \Nn$, and let
  $\m \in \N^P$ be such that $L(\m) = F(\m)$. It is the case that $L(c
    \cdot \m) = F(c \cdot \m) = F(\m)$.
\end{lemma}

\begin{proof}
  We first show that $F(c \cdot \m) = F(\m)$, and then that $L(c \cdot
    \m) = F(c \cdot \m)$.

  We trivially have $F(c \cdot \m) \supseteq F(\m)$. For the sake of
  contradiction, suppose there exists a transition $t \in
    F(c \cdot \m)$ such that $t \notin F(\m)$. Let $\sigma_1$ be a run
  such that $c \cdot \m \trans{\sigma_1 t}$. Without loss of
  generality, we may assume that $\parikh{\sigma_1} \subseteq
    F(\m)$. Indeed, if there is some $t' \in \parikh{\sigma_1}$ such
  that $t' \notin F(\m)$, then we can shorten $\sigma_1$ and take the
  shortened run which enables $t'$ instead.

  Let $\sigma_1 = t_1 t_2 \cdots t_n$. Recall that $t_i \in L(\m) =
    F(\m)$ for each $t_i$, that is, from any marking reachable from
  $\m$, we can reach a marking that enables $t_i$. Therefore, we can
  define a run $\sigma_2 \defeq \phi_1 t_1 \phi_2 t_2 \cdots \phi_n
    t_n$, where $\phi_i$ is a run from $\m + \effect{\phi_1
      t_1 \cdots \phi_{i-1} t_{i-1}}$ that enables $t_i$.

  If there exists a transition $s$ in the run $\sigma_2$ such that $\supppre{s}
    \cap \supppre{t} \neq \emptyset$, then $\pre{s} = \pre{t}$ as $\pn$
  is free-choice. Hence, since $s \in F(\m)$, we obtain $t \in F(\m)$,
  which is a contradiction. Thus no transition in $\sigma_2$ can consume tokens from places in $\supppre{t}$.
  Since $c \cdot \m \trans{\sigma_1 t}$, we know that
  \begin{align*}
    \supppre{t} \subseteq  \support{c \cdot\m} \cup \bigcup_{i = 1}^n \supppost{t_i} = \support{\m} \cup \bigcup_{i = 1}^n \supppost{t_i}.
  \end{align*}
  Altogether, this means that the transitions $t_i$ put enough tokens such that all places in $\supppre{t}$ are marked, and that $\sigma_2$ cannot consume any of these tokens.
  Therefore,
  $\m \trans{\sigma_2 t}$, which is a contradiction.

  It remains to prove that $L(c \cdot \m) = F(c \cdot \m)$. We have
  $L(\m) \subseteq L(c \cdot \m) \subseteq F(c \cdot \m)$. Since
  $L(\m) = F(\m) = F(c \cdot \m)$, these inclusions are in fact
  equalities, and we are done. \qed
\end{proof}

\begin{lemma}
  Let $\pn = (P, F, T)$ be a free-choice workflow net, let $c \in \Nn$
  and let $\m \in \N^P$ be such that $L(c \cdot \m) =
    F(c \cdot \m)$. If $c \cdot \m \reach \fmarked{c}$ and $(\pn,
    c \cdot \m)$ is bounded, then $\m = \fmarked{1}$.
\end{lemma}

\begin{proof}
  Recall that no transition of a workflow net consumes from $\output$,
  \ie\ $\pre{t}[\output] = 0$ for all $t \in T$. Thus, we either have
  $\m[\f] = 0$ or $\m[\f] = 1$.

  If $\m[\f] = 0$, then there is some transition $t \in F(c \cdot \m)$
  such that $f \in \post{t}[\output] > 0$. Since $t \in L(c \cdot
    \m)$, it follows that from $c \cdot \m$, we can reach $\m'$ with
  $\m'[\f]$ abritrarily large, as $t$ puts a token into $\f$ and can
  be fired arbitrarily often from $c \cdot \m$. This contradicts the
  fact that $(\pn, c \cdot \m)$ is bounded. Hence, $\m[\output] =
    1$. We can write $\m$ as $\m = \fmarked{1} + \m'$ where
  $\m'[\output] = 0$. We have $c \cdot \m = \fmarked{c} + c \cdot
    \m'$. If $\m' = \vec{0}$, then we are done.
  Otherwise, we obtain a contradiction. Indeed, it
  cannot be the case that $\fmarked{c} + c \cdot \m' \reach
    \fmarked{c}$, as every transition of a workflow net produces at
  least one token (and none consumes from $\output$). \qed
\end{proof}

\lemContsoundBounded*

\begin{proof}
  Assume for contradiction that there exists $k \in \Nn$ such that
  $(\pn, \imarked{k})$ is unbounded, but $\pn$ is continuously
  sound. There exist marking $\m$ and $\m' > \m$ such that
  $\imarked{k} \reach \m \reach \m'$. By
  \autoref{lem:cont-equiv-reach}, we have $\imarked{1} \creach \n
    \creach \n'$, where $\n \defeq (1/k) \cdot \m$ and $\n' \defeq (1/k)
    \cdot \m'$. As $\pn$ is continuously sound, it must hold that $\n
    \creach \fmarked{1}$. It follows that
  \[
    \n' = \n + (\n' - \n) \creach \fmarked{1} + (\n' - \n).
  \]
  This contradicts the assumption that $\pn$ is continuously sound, as
  each transition of a workflow net produces at least one token, and
  none consumes from $\f$. \qed
\end{proof}

\subsection{Missing definition of the arc weight encoding of \Cref{sec:experimental}}\label{ssec:arc:weights}

Recall that under our definition, Petri nets do not
have arc weights as $F \colon ((P \times T) \cup (T \times P))
  \to \set{0, 1}$. Petri nets with arc weights are defined exactly as
Petri nets but with $F \colon ((P \times T) \cup (T \times P)) \to
  \N$. An example of the arc weight encoding described in the main text
is shown in \Cref{fig:encoding}.

\begin{figure}[h]
    \centering

    \begin{tikzpicture}[node distance=1cm, scale=0.9, transform shape]
        \colorlet{colFont}{black!75}

        \tikzstyle{cplace} = [place, font=\scriptsize, text=colFont,
        minimum size=15pt];
        \tikzstyle{ctransition} = [transition, font=\scriptsize, text=colFont,
        minimum height=12pt, minimum width=15pt];

        \node[cplace] (i) at (0.25,1) {$s$};
        \node[cplace, right=of i] (l0) {$l_1$};

        \node[ctransition, red] at($(i)!0.5!(l0)$) (ti) {$t_s$};

        \path[->]
        (i) edge[]            node[] {}           (ti);

        \node[cplace, right=of l0] (r0) {$r_1$};

        \path[->]
        (ti) edge[]            node[] {}           (l0);

        \node[ctransition, red] (t0r) at ($(l0)!0.5!(r0)$) {$t_r$};

        \path[->]
        (l0) edge[]            node[] {}           (t0r)
        (t0r) edge[]            node[] {}           (r0)
        ;

        \node[cplace, above=1.5cm of l0] (l1) {$l_2$};
        \node[cplace, above=1.5cm of r0] (r1) {$r_2$};

        \node[ctransition, red] (t1l) at ($(l0)!0.5!(l1)$) {$t_{2,l}$};
        \node[ctransition, red] (t1r) at ($(r0)!0.5!(r1)$) {$t_{2,r}$};

        \path[->]
        (l0) edge[]            node[] {}           (t1l)
        (r0) edge[]            node[] {}           (t1l)
        (t1l) edge[]            node[] {}           (l1)
        ;

        \path[->]
        (l0) edge[]            node[] {}           (t1r)
        (r0) edge[]            node[] {}           (t1r)
        (t1r) edge[]            node[] {}           (r1)
        ;

        \node[cplace, above=1.5cm of r1] (r2) {$r_3$};

        \node[ctransition, red] (t2r) at ($(r1)!0.5!(r2)$) {$t_{3,r}$};

        \path[->]
        (l1) edge[]            node[] {}           (t2r)
        (r1) edge[]            node[] {}           (t2r)
        (t2r) edge[]            node[] {}           (r2)
        ;

        \node[ctransition, right=1.5cm of r1] (t) {$t$};

        \path[->]
        (r2) edge[bend left]            node[] {}           (t)
        (r0) edge[bend right]            node[] {}           (t)
        ;

        \node[cplace, right=of t] (d2) {$d_2$};
        \node[cplace, below=1.5cm of d2] (d1) {$d_1$};
        \node[cplace, above=1.5cm of d2] (d3) {$d_3$};
        \node[cplace, right=of d2] (h2) {$h_2$};

        \node[ctransition] at ($(d2)!0.5!(d3)$) (t3d) {$t_{3,d}$};

        \node[cplace, right=of d1] (h1) {$h_1$};
        \node[ctransition] at ($(d1)!0.5!(d2)$) (t2d) {$t_{2,d}$};
        \node[ctransition] at ($(h2)!0.5!(h1)$) (t2h) {$t_{2,h}$};
        \node[ctransition] at ($(d1)!0.5!(h1)$) (t1h) {$t_{h}$};

        \path[->]
        (d3) edge[]            node[] {}           (t3d)
        (t3d) edge[]            node[] {}           (d2)
        (t3d) edge[]            node[] {}           (h2)
        ;

        \path[->]
        (h2) edge[]            node[] {}           (t2h)
        (t2h) edge[]            node[] {}           (d1)
        (t2h) edge[]            node[] {}           (h1)
        ;

        \path[->]
        (t) edge[]            node[] {}           (d3)
        (t) edge[]            node[] {}           (d2)
        ;

        \path[->]
        (d2) edge[]            node[] {}           (t2d)
        (t2d) edge[]            node[] {}           (d1)
        (t2d) edge[]            node[] {}           (h1)
        ;

        \path[->]
        (d1) edge[]            node[] {}           (t1h)
        (t1h) edge[]            node[] {}           (h1)
        ;

        \node[cplace, right= of h1] (f) {$f$};
        \node[ctransition] at ($(h1)!0.5!(f)$) (tf) {$t_f$};

        \path[->]
        (h1) edge[]            node[] {}           (tf)
        (tf) edge[]            node[] {}           (f)
        ;

        \colorlet{colOrig}{blue}
        \tikzstyle{dplace} = [place, colOrig, minimum size=10pt, text=colOrig];
        \tikzstyle{dtransition} = [transition, colOrig, minimum size=10pt, text=colOrig];

        \node[dplace, left=2.5cm of r2, label=right:\textcolor{colOrig}{$f$}, yshift=-10pt] (forig) {};
        \node[dtransition, left=0.5cm of forig, label=below:\textcolor{colOrig}{$t$}] (torig) {};
        \node[dplace, left=0.5cm of torig, label=left:\textcolor{colOrig}{$s$}] (sorig) {};

        \path[colOrig, ->, auto]
        (sorig) edge node {$5$} (torig)
        (torig) edge node {$6$} (forig)
        ;
    \end{tikzpicture}
    \caption{\emph{Top left (in \textcolor{blue}{blue}):} A Petri net $\pn$ with arc weights.
        \emph{Center:} A Petri net $\encpn$ without arc weights
        that simulates behaviour of $\pn$.
        For each transition colored in \textcolor{red}{red}, the reverse transition is also part of
        $\encpn$, and is merely not drawn to avoid overcrowding the figure.
        For ease of presentation, places and transitions of $\encpn$ contain their names (not values).
    }
    \label{fig:encoding}
\end{figure}

In this section, we will use $t^{-1}$ to denote the reverse transition
of transition $t$, as done in the coNP membership proof of
\Cref{thm:conp}.

Formally, to simulate a transition $t$, we
add places $P_{p,t}$ and transitions $T_{p,t}$ for each place $p$ with $b \defeq \pret[p] > 1$, and
places $P'_{p,t}$ and transitions $T'_{p,t}$ for each place $p$ with $b' \defeq \postt[p] > 1$.

From now on, when we define a transition $t$, we assume that $\pre{t}[p'] = 0$ and $\post{t}[p'] = 0$
for each place $p'$ except those given explicitly.
We define $P_{p,t}$ as follows.
We denote by $b_1,b_2,\dots,b_n$ the binary representation of $b$, that is, $b = \sum_{i=1}^{n} b_i \cdot 2^{i-1}$, and
similarly $b'_1,b'_2,\dots,b'_{n'}$ for $b'$. The set $P_{p,t}$ consists of $2n-1$ places. For every $i \in [1..n-1]$, we add two places $l_i$ and $r_i$; and an additional place $r_n$.
The set $T_{p,t}$ contains the following transitions:
\begin{itemize}
  \item $t_p$, where $\pre{t_p}[p] = \post{t_p}[l_1] \defeq 1$;

  \item $t_{r}$, as well as its reverse $t^{-1}_{r}$, where $\pre{t_r}[l_1] = \post{t_r}[r_1] \defeq 1$;

  \item for each $i \in [2..n-1]$ the transitions $t_{i,l}$, $t_{i,r}$ and their reverses $t_{i,l}^{-1}$, $t_{i,r}^{-1}$,
        where $\pre{t_{i,l}}[r_{i-1}] = \pre{t_{i,l}}[l_{i-1}] \defeq 1$,
        $\pre{t_{i,r}} = \pre{t_{i,l}}$, and
        $\post{t_{i,l}}[l_i] = \post{t_{i,r}}[r_i] \defeq 1$,

  \item the transition $t_{n,r}$ and its reverse $t_{n,r}^{-1}$, where $\pre{t_{n,r}}[l_{n-1}] = \pre{t_{n,r}}[r_{n-1}] \defeq 1$
        and
        $\post{t_{n,r}}[r_n] \defeq 1$.
\end{itemize}
We further redefine $t$ to have $\pre{t}[p] \defeq 0$
and $\pre{t}[r_{i}] \defeq 1$ for all $i$ such that $b_{i} = 1$.

The set $P'_{p,t}$ consists of $2n' - 1$ places.
We have $d_i$ and $h_i$
for each $i \in [1..n'-1]$,
and an additional place $d_n'$.
The set $T'_{p,t}$ contains the following transitions:
\begin{itemize}
  \item $t_p$, where $\pre{t_p}[h_1] = \post{t_p}[p] \defeq 1$,

  \item $t_{1,h}$, where $\pre{t_{1,h}}[d_1] = \post{t_{1,h}}[h_1] \defeq 1$,

  \item for each $i \in [2..n-1]$, the transitions $t_{i,d}$ and $t_{i,h}$, where $\pre{t_{i,d}}[d_i] =
          \pre{t_{i,h}}[h_i] \defeq 1$, $\post{t_{i,d}}[d_{i-1}] = \post{t_{i,d}}[h_{i-1}] \defeq 1$,
        and $\post{t_{i,h}} \defeq \post{t_{i,d}}$.
\end{itemize}
We further redefine $t$ to have $\post{t}[p] \defeq 0$
and $\post{t}[d_i] \defeq 1$ for each $i$ such that $b'_{i} = 1$.

Given a Petri net $\pn = (P,T,F)$, let us denote by $\encpn = (P',T',F')$
the transformed $\pn$ where all transitions with arc weights are modified by the gadget defined above.
To avoid any confusion, we denote markings in $\pn$ as $\m$ and $\m'$, and markings in $\encpn$ as $\n$ and $\n'$.
As $\encpn$ does not remove (but only adds) places, we may treat markings on $\pn$ as markings on $\encpn$,
where all places in $P' \setminus P$ are marked with zero token.

Recall that $\parikh{\sigma}$ is a vector mapping each transition $t$ to the number of times $t$ is used in run $\sigma$. In the following, let $p \in P$ and $t \in T$ be such that $\pre{t}[p] = b \geq 1$.
Let $b_1, \dots, b_n$ be the binary representation of $b$. Furthermore, let $P_{p,t}$ and $T_{p,t}$ be defined as above.

We are ready to state some helpful lemmas.

\begin{lemma}\label{lem:elevator-up}
  Let $i \in [1..n]$. We have $\{p\colon 2^{i-1}\}
    \trans{\sigma} \{r_i\colon 1\}$ in $\encpn$ with
  $\support{\parikh{\sigma}} \subseteq T_{p,t}$. Further, if $i <
    n$, then $\{p\colon 2^{i-1}\} \trans{\sigma'} \{l_i\colon 1\}$ in
  $\encpn$ with $\support{\parikh{\sigma'}} \subseteq T_{p,t}$.
\end{lemma}

\begin{proof}
  We proceed by induction.
  For $i = 1$, we have
  \begin{align*}
    \{p \colon 2^{1-1}\} = \{p \colon 1\} \trans{t_p} \{l_1\colon 1\}
    \trans{t_{r}}                                   \{r_1\colon 1\}.
  \end{align*}

  For $i > 1$, we have
  $\{p\colon 2^{i-1}\} = \{p \colon 2^{i-2} + 2^{i-2}\} \reach \{r_{i-1}\colon 1, l_{i-1}\colon 1\}$
  by the induction hypothesis.
  Thus, we have $\{r_{i-1}\colon 1, l_{i-1}\colon 1\} \trans{t_{i,r}} \{r_{i}\colon 1\}$.
  If $i < n$, then we additionally have $\{r_{i-1}\colon 1, l_{i-1}\colon 1\} \trans{t_{i,l}} \{l_i\colon 1\}$.
  We conclude the proof by pointing out that for all $i$, $t_p, t_r, t_{i,r}, t_{i,l} \in T_{p,t}$. \qed
\end{proof}

The proof of the lemma below follows by the fact that all transitions
of $T_{p,t}$ are reversible.

\begin{lemma}\label{lem:elevator-down}
  Let $i \in [1..n]$. We have $\{r_i\colon 1\} \trans{\sigma}
    \{p\colon 2^{i-1}\}$ in $\encpn$ with $\support{\parikh{\sigma}}
    \subseteq T_{p,t}$. Further, if $i < n$, then $\{l_i\colon 1\}
    \trans{\sigma'} \{p\colon 2^{i-1}\}$ in $\encpn$ with
  $\support{\parikh{\sigma'}} \subseteq T_{p,t}$.
\end{lemma}

For the next lemma, let  $p \in P$ and $t \in T$ be such that $\post{t}[p] = b \geq 1$.
Let $b_1,\dots,b_n$ be the binary representation of $b$. Let $P'_{p,t}$ and $T'_{p,t}$ be as defined above.

\begin{lemma}\label{lem:ladder-down}
  Let $i \in [1..m]$. We have $\set{d_i \colon 1} \trans{\sigma}
    \set{p \colon 2^{i-1}}$ in $\encpn$ with $\support{\parikh{\sigma}}
    \subseteq T'_{p,t}$. Further, if $i < n$, then $\set{h_i\colon 1}
    \trans{\sigma'} \set{p \colon 2^{i-1}}$ in $\encpn$ with
  $\support{\parikh{\sigma'}} \subseteq T'_{p,t}$.
\end{lemma}

\begin{proof}
  We proceed by induction on $i$.
  If $i = 1$, then we have $2^{i-1} = 1$ and hence $\set{d_1 \colon 1} \trans{t_{h}} \set{h_1 \colon 1} \trans{t_{p}} \set{p \colon 1}$.

  For $i > 1$, we have $\set{d_i\colon 1} \trans{t_{i,d}} \set{d_{i-1}\colon 1, h_{i-1}\colon 1}$.
  If $i < n$, then we additionally have $\set{h_i\colon 1} \trans{t_{i,h}} \set{d_{i-1}\colon 1, h_{i-1}\colon 1}$.
  It follows from the induction hypothesis that
  $\set{d_{i-1}\colon 1, h_{i-1}\colon 1} \trans{\sigma\sigma'} \set{p\colon 2^{i-2} + 2^{i-2}} = \set{p \colon 2^{i-1}}$.
  We conclude by pointing out that, for all $i$, we have $t_{p}, t_h, t_{i,d}, t_{i,h} \in T'_{p,t}$. \qed
\end{proof}

\begin{definition}\label{def:pinv}
  Let $U \subseteq T$.
  A vector $\vec{x} \colon P \to \Q$ is a \emph{place invariant}
  over $U$ if the
  following holds for all $t \in U$:
  \begin{equation}
    \sum_{p \in P} \pre{t}[p] \cdot \vec{x}[p] = \sum_{p \in P}
    \post{t}[p] \cdot \vec{x}[p].\label{eq:pinv}
  \end{equation}
\end{definition}

\begin{proposition}[adapted from {\cite[Prop.~2.27]{desel1995free}}]
  \label{prop:fundinv}
  Let $U \subseteq T$ and let $\vec{x}$ be a place invariant over $U$.
  If $\m \trans{\sigma} \m'$ with $\support{\parikh{\sigma}} \subseteq U$, then $\vec{x} \cdot \m = \vec{x} \cdot \m'$.
\end{proposition}

Let us define the vector $I_{p,t}$ with
$I_{p,t}[p] \defeq 1$, $I_{p,t}[r_i] \defeq 2^{i-1}$ and $I_{p,t}[l_i] \defeq 2^{i-1}$,
where $r_i$ and $l_i$ are the places previously defined in $P_{p,t}$.
It is easy to see that $I_{p,t}$ is a place invariant of $T_{p,t}$.

Let $R \defeq \{t \in T \mid \pre{t}[p] \geq 2\}$ and $S \defeq \{t \in T \mid \post{t}[p] \geq 2\}$.
We further define the vector $I_{p} \colon \set{p} \cup \bigcup_{t \in R} P_{p,t} \cup \bigcup_{t \in S} P'_{p,t} \to \Q$,
where $P_{p,t} = \emptyset$ if $\pret[p] \leq 1$ and $P'_{p,t} = \emptyset$ if $\postt[p] \leq 1$.
We define $I_p[p] \defeq 1$ and $I_p[p'] \defeq 2^{i-1}$ if $p' \in \{r_i,l_i,d_i,h_i\}$ for some $i$.
Note that this is well-defined by our choice of domain of $I_p$.
It is easy to convince oneself that $I_{p}$ is a
place invariant of $T' \setminus T$.


\newcommand{\Ppre}{G}
\newcommand{\Ppost}{H}
\newcommand{\bin}[1]{n(#1)}

We introduce some notation.
For a transition $t \in T$,
let $\Ppre \defeq \{p \in P \mid \pre{t}[p] \geq 2\}$ and $\Ppost \defeq \set{p \in P \mid \post{t}[p] \geq 2}$.
For a place $p \in \Ppre$, we write $b(p) \defeq \pre{t}[p]$.
For $i \in \N$, we write $\bin{i}$ to denote the number of bits in the binary representation of $i$.
Let $b_1(p),\dots,b_{\bin{b(p)}}(p)$ denote the bits of the binary representation of $b(p)$.
Let $r_i(p)$ denote the place $r_i$ in $P_{p,t}$.
Similarly, given $p \in \Ppost$, we write $c(p) \defeq \post{t}[p]$,
we let $c_1(p),\dots,c_{\bin{c(p)}}(p)$ be the bits of the binary representation of $c(p)$,
and we further write $d_i(p)$ to denote the place $d_i$ of $P'_{p,t}$.
In the following, we denote by $t$ the transition in $\pn$, and by $t'$ the corresponding transition in $\encpn$.

\begin{lemma}\label{lem:preservesreach}
  Let $t \in T$ and let $\m,\m'$ be markings of $\pn$
  with $\m' = \m + \effect{t}$. It holds that $\m \trans{t} \m'$ in $\pn$ iff $\m \trans{\pi t' \pi'} \m'$ in $\encpn$,
  where $\support{\parikh{\pi}} \subseteq \bigcup_{p \in \Ppre} T_{p,t}$ and $\support{\parikh{\pi'}} \subseteq \bigcup_{p \in \Ppost} T'_{p,t}$.
\end{lemma}

\begin{proof}
  $\Rightarrow$) By definition of $\encpn$, $\m[p] \geq \pre{t'}[p]$ for all
  $p \in P \setminus \Ppre$.
  By definition of $\encpn$, it holds that
  $\pre{t'}[r_i(p)] = b_i(p)$ for all $i \in [1..\bin{b(p)}]$.
  Note that $\m[p] \geq b(p) = \sum_{i = 1}^{\bin{b(p)}} 2^{i-1} \cdot b_i(p)$.
  Thus, it follows from \Cref{lem:elevator-up} that
  $\{p\colon b(p)\} \trans{\sigma} \sum_{i=1}^{\bin{b(p)}} \{r_i(p) \colon b_i(p)\}$.
  So, in particular, \[\m[p] \trans{\sigma} \m - \{p \colon b(p) \} + \sum_{i=1}^{\bin{b(p)}} \{r_i(p)\colon b_i(p)\}.\]

  Since the transitions in $\sigma$ do not have an effect on places other than $P_{p,t} \cup \{p\}$,
  we can invoke \Cref{lem:elevator-up} individually for each $p \in \Ppre$, and thus obtain
  \[\m \trans{\sigma_1 \cdots \sigma_{\abs{\Ppre}}} \m + \sum_{p \in \Ppre} \sum_{i=1}^{\bin{b(p)}} \{r_i(p) \colon b_i(p)\} - \{p \colon b(p)\},\]
  where $\support{\parikh{\sigma_1}},\dots,\support{\parikh{\sigma_{\abs{\Ppre}}}} \subseteq \bigcup_{p \in \Ppre} T_{p,t}$.
  By definition, $t'$ is enabled in this marking and its firing leads to
  \begin{multline*}
    \m   - \sum_{p \in \Ppre} \{p \colon b(p)\} - \sum_{p \in P \setminus \Ppre} \set{p \colon \pre{t}[p]}
    +
    \sum_{p \in P \setminus \Ppost} \set{p \colon \post{t}[p]} + \sum_{p \in \Ppost} \sum_{i = 1}^{\bin{c(p)}} \set{d_i(p) \colon c_i(p)} \\
    = \m - \pre{t} + \sum_{p \in P \setminus \Ppost} \set{p \colon \post{t}[p]} + \sum_{p \in \Ppost} \sum_{i = 1}^{\bin{c(p)}} \set{d_i(p) \colon c_i(p)}.
  \end{multline*}
  Let us denote the latter marking as $\m'$.
  By invoking \Cref{lem:ladder-down} individually on each $d_i(p)$,
  it follows that for each $p \in \Ppost$:
  \begin{align*}
    \m' \trans{\sigma'_{1}\cdots\sigma'_{\abs{\Ppost}}}\  & \m - \pre{t} + \sum_{p \in P \setminus \Ppost} \set{p \colon \post{t}[p]} + \sum_{p \in \Ppost} \sum_{i = 1}^{\bin{c(p)}} \set{p\colon
      2^{i-1} c_i(p)} =                                                                                                                                                                            \\
                                                          & \m - \pre{t} + \sum_{p \in P \setminus \Ppost} \set{p \colon \post{t}[p]} + \sum_{p \in \Ppost} \set{p\colon
      \postt[p]}=                                                                                                                                                                                  \\
                                                          & \m - \pre{t} + \post{t} = \m + \effect{t}.
  \end{align*}
  We conclude this direction by noting that $\support{\parikh{\sigma_{1}}},\dots,\support{\parikh{\sigma_{\abs{\Ppost}}}} \subseteq \bigcup_{p \in \Ppost} T'_{p,t}$
  by \Cref{lem:ladder-down}.

  $\Leftarrow$) We have $\m \trans{\sigma t' \sigma'} \m'$.
  Let us denote by $\m_{1}$ the marking
  such that $\m \trans{\sigma} \m_{1}$.
  It must be the case that
  \begin{align*}
    \m_{1} \geq \pre{t'} = & \sum_{p \in P \setminus \Ppre} \pre{t}[m] + \sum_{p \in \Ppre} \sum_{i \in \bin{b(p)}} \set{r_i(p)\colon b_{i}(p)}.
  \end{align*}
  Recall that for each $p \in \Ppre$, $I_{p,t}$ is a place invariant of $T_{p,t}$.
  In particular, among transitions from $T' \setminus T$, places in $\{p\} \cup P_{p,t}$
  are only affected by transitions in $T_{p,t}$.
  So, $I_{p,t} \cdot \m = I_{p,t} \cdot \m_{1}$ by \cref{prop:fundinv}.
  Since $\m_{1} \geq \sum_{i \in \bin{b(p)}} \set{r_i(p)\colon b_{i}(p)}$,
  we have $I_{p,t} \cdot \m_{1} \geq \sum_{i \in \bin{b(p)}} 2^{i-1} b_{i}(p)$.
  Thus, the same must hold for $I_{p,t} \cdot \m$.
  But among places in $\{p\} \cup P_{p,t}$, $\m$ marks only $p$, as it is (by projection) a marking of $\pn$.
  Since $I_{p,t}[p] = 1$, it must hold that $\m[p] \geq \sum_{i \in \bin{b(p)}} 2^{i-1} b_{i}(p)
    = b(p)$, where the last equality follows from the fact that $b_{1}(p),\dots,b_{\bin{b(p)}}(p)$ is the binary representation of
  $b(p)$. So, $\m[p] \geq \pre{t}[p]$ holds by definition of $b(p)$.
  Therefore, $\m$ enables $t$,
  and consequently $\m \trans{t} \m + \effect{t} = \m'$, and we are done. \qed
\end{proof}

\begin{lemma}\label{lem:unreach-preserves}
  Let $\m, \m'$ be markings of $\pn$.
  If $\m \trans{\sigma} \m'$ in $\encpn$ and $\support{\parikh{\sigma}} \subseteq T' \setminus T$, then $\m = \m'$.
\end{lemma}

\begin{proof}
  We argue for each place $p \in P$ individually that
  $\m[p] = \m'[p]$.

  Recall that $I_p$ is a place invariant over $T' \setminus T$.
  Therefore, $I_p \cdot \m = I_p \cdot \m'$ by \Cref{prop:fundinv}.
  Note also that in the domain of $I_p$, the only place in $P$ is $p$.
  Since $\m$ and $\m'$ are markings of $\pn$, and consequently all
  places in the domain of $I_p$ other than $p$ must be unmarked,
  it follows that $I_p[p] \cdot \m[p] = I_p[p] \cdot \m'[p]$. Thus, $\m[p] = \m'[p]$. \qed
\end{proof}

\begin{lemma}\label{lem:runreach}
  Let $\m, \m'$ be markings of $\pn$.
  If $\m \reach \m'$ in $\encpn$,
  then $\m \reach \m'$ in $\pn$.
\end{lemma}

\begin{proof}
  Let $\m \trans{\pi} \m'$.
  If $\support{\parikh{\pi}} \subseteq T' \setminus T$, then $\m = \m'$ by \Cref{lem:unreach-preserves}, and we are done. So, assume that $t \in T$ for some $t \in \pi$. We factor run $\pi$ so that $\pi = \sigma_1 t_1 \sigma_1' \cdots \sigma_n t_n \sigma'_n$
  with $t_1, \dots, t_n \in T$ and \[\support{\parikh{\sigma_1}},\support{\parikh{\sigma'_1}},\dots,\support{\parikh{\sigma_n}},\support{\parikh{\sigma'_n}} \subseteq T \setminus T'.\]
  It follows from \Cref{lem:preservesreach} that $\imarked{1} \trans{t_1 \cdots t_n} \fmarked{k}$. \qed
\end{proof}

\begin{proposition}
  For any workflow net $\pn$ and any $k \in \Nn$,
  $\pn$ is $k$-quasi-sound iff $\encpn$ is $k$-quasi-sound.
\end{proposition}

\begin{proof}
  This follows immediately from \Cref{lem:preservesreach,lem:runreach}. \qed
\end{proof}

\begin{proposition}
  For any workflow net $\pn$ and any $k \in \Nn$,
  $\pn$ is $k$-sound iff $\encpn$ is $k$-sound.
\end{proposition}

\begin{proof}
  $\Rightarrow$)
  Assume $\pn$ is $k$-sound.
  Let $\m$ be a marking of $\encpn$
  such that $\imarked{k} \reach \m$ in $\encpn$.
  If $\m$ is also a marking of $\pn$,
  then $\imarked{k} \reach \m$ in $\pn$ by \Cref{lem:runreach}.
  Thus, $\m \reach \fmarked{k}$ in $\pn$ by $k$-soundness,
  and $\m \reach \fmarked{k}$ in $\encpn$ by \Cref{lem:preservesreach}.
  If $\m$ is not a marking on $\pn$, then, for each place $p \in P' \setminus P$,
  we can invoke \Cref{lem:elevator-down,lem:ladder-down} in order to
  obtain a marking $\m'$ which marks only places in $P$.
  So, we have $\imarked{k} \reach \m \reach \m'$ in $\encpn$, and it follows by \Cref{lem:runreach}
  that $\imarked{k} \reach \m'$ in $\pn$.
  Thus, $\m' \reach \fmarked{k}$ in $\pn$ by $k$-soundness, and $\m' \reach \fmarked{k}$ in $\encpn$ by \Cref{lem:preservesreach},
  which shows that $\encpn$ is $k$-sound.

  $\Leftarrow$)
  Assume $\encpn$ is $k$-sound.
  Let $\m$ be a marking of $\pn$ such that
  $\imarked{k} \reach \m$ in $\pn$.
  If follows from \Cref{lem:preservesreach}
  that $\imarked{k} \reach \m$ in $\encpn$.
  By $k$-soundness of $\encpn$, we have $\m \reach \fmarked{k}$ in $\encpn$.
  Thus, $\m \reach \fmarked{k}$ in $\pn$ by \Cref{lem:runreach}. \qed
\end{proof}

\subsection{Missing proofs of \Cref{sec:experimental}}

Let us prove the properties claimed about the instances of
\Cref{fig:synthetic}.

\begin{restatable}{proposition}{propThreeSynthetic}\label{prop:three:synthetic}
  It is the case that
  \begin{enumerate}
    \item $\pn_c$ is $c$-unsound and $k$-sound for all $k \in
            [1..c-1]$.\label{itm:gensound}
    \item $\pn_{\text{sound-}c}$ is $kc$-sound for all $k \in
            \Nn$,\label{itm:sound}

    \item $\pn_{\neg \text{quasi-}c}$ is not structurally quasi-sound,
          and\label{itm:nquasi}

    \item $\pn_{\neg \text{sound-}c}$ is $(m c)$-quasi-sound for all $m \in
            \Nn$, not $k$-quasi-sound for any other number $k \in \Nn$, and
          not structurally sound.\label{itm:nsound}
  \end{enumerate}
\end{restatable}

\begin{proof}
  \noindent\emph{Items~\ref{itm:sound} and~\ref{itm:nquasi}}. They
  follow from the definitions of the unique transition.

  \medskip\noindent\emph{Item~\ref{itm:gensound}.}
  We first focus on $k$-soundness. Let $k \in [1..c-1]$ and let $\m$
  be a marking such that $\imarked{k} \reach \m$. We must show that
  $\m \reach \fmarked{k}$.

  Recall the definition of a place invariant from \Cref{def:pinv}.

  Let $\vec{x}[\initial] \defeq c + 1$, $\vec{x}[p] \defeq 1$,
  $\vec{x}[r] \defeq c$ and $\vec{x}[\output] \defeq c + 1$. It is
  readily seen that $\vec{x}$ is a place invariant.
  Recall \Cref{prop:fundinv}: for any two markings $\n$ and $\n'$, if $\n \reach
    \n'$, then $\vec{x} \cdot \n = \vec{x} \cdot \n'$. Since $\imarked{k}
    \reach \m$, we have $\vec{x} \cdot \imarked{k} = (c + 1) \cdot k =
    \vec{x} \cdot \m$.

  From marking $\m$, transition $t_{\initial}$ can be fired
  $\m[\initial]$ times, which leads to marking
  \[
    \m_1 \defeq \{p \colon \m[p] + (c+1) \cdot \m[\initial], r \colon
    \m[r], \output \colon \m[\output]\}.
  \]
  From $\m_1$, transition $t_r$ can be fired $\m_1[\initial] \div c$
  times, which leads to marking
  \[
    \m_2\defeq \{p \colon \m_1[p]~\mathrm{mod}~c, r \colon \m_1[r] +
    \m_1[p] \div c, \output \colon \m_1[\output]\}.
  \]
  Recall that from place invariant $\vec{x}$, we have
  \[
    (c + 1) \cdot k
    = (c + 1) \cdot \m[\initial] + \m[p] + c \cdot \m[r] +
    (c + 1) \cdot \m[\output].
  \]
  By reorganizing this equation, we obtain
  \begin{equation}
    \m[p] + \m[i] =
    (c + 1)(k - \m[\output]) - c \cdot (\m[\initial] + \m[r]).\label{eq:ieq}
  \end{equation}
  This means that
  \begin{align}
    \m_2[p]
     & = \m_1[p]~\mathrm{mod}~c
     &                                                & \text{(by def.\ of $\m_2$)}\nonumber      \\
     & = (\m[p] + (c + 1) \cdot \m[i])~\mathrm{mod}~c
     &                                                & \text{(by def.\ of $\m_1$)}\nonumber      \\
     & = (\m[p] + \m[i])~\mathrm{mod}~c\nonumber                                                  \\
     & = k - \m[\output]
     &                                                & (\text{by~\eqref{eq:ieq}}).\label{eq:m2p}
  \end{align}
  Since $\imarked{k} \reach \m_2$, from place invariant $\vec{x}$, we
  obtain
  \[
    (c + 1) \cdot k
    = (c + 1) \cdot \m_2[\initial] + \m_2[p] + c \cdot \m_2[r] +
    (c + 1) \cdot \m_2[\output].
  \]
  By reorganizing this equation, we obtain
  \begin{equation}
    c \cdot \m_2[r] =
    (c + 1) \cdot (k - \m_2[\initial] - \m_2[\output]) -
    \m_2[p].\label{eq:ieq2}
  \end{equation}
  This means that
  \begin{align*}
    c \cdot \m_2[r]
     & = (c + 1) \cdot (k - \m_2[\initial] - \m_2[\output]) - \m_2[p]
     &                                                                & \text{(by~\eqref{eq:ieq2})}                    \\
     & = (c + 1) \cdot (k - \m_1[\output]) - (k - \m[\output])
     &                                                                & \text{(by def.\ of $\m_2$ and~\eqref{eq:m2p})} \\
     & = (c + 1) \cdot (k - \m[\output]) - (k - \m[\output])
     &                                                                & \text{(by def.\ of $\m_1$)}                    \\
     & = c \cdot (k - \m[\output]).
  \end{align*}

  Altogether, we have $\m_2[r] = (k - \m[\output]) = \m_2[p]$. Thus,
  from $\m_2$, transition $t_{\output}$ can be fired ${k -
        \m[\output]}$ times, which leads to marking
  \[
    \fmarked{\m_1[\output] + (k - \m[\output])}
    = \fmarked{\m[\output] + k - \m[\output]}
    = \fmarked{k}.
  \]
  This concludes the proof of $k$-soundness as $\m \reach \m_1 \reach
    \m_2 \reach \fmarked{k}$.

  It remains to consider the case where $k = c$. We have
  \begin{align*}
    \imarked{c} \trans{t_{\initial}^{c}} \{p \colon (c+1)\cdot c\}
    \trans{t_r^{c+1}} \{r\colon (c+1)\}.
  \end{align*}
  No transition is enabled in the latter marking. So, we have $\{r
    \colon (c+1)\} \not\reach \fmarked{c}$ and hence $c$-unsoundness
  follows. We are done proving this item.

  \medskip\noindent\emph{Item~\ref{itm:nsound}}.
  Let $k \in \Nn$ be a number that is not a multiple of $c$. Let us
  first show that $\imarked{k} \not\reach \fmarked{k}$. For the sake
  of contradiction, assume there exists a run $\rho$ such that
  $\imarked{k} \trans{\rho} \fmarked{k}$. Note that $\rho$ needs to
  fire $t_{\initial}$ exactly $k$ times, since no other transition
  consumes from $\initial$. Without loss of generality, let us reorder
  $\rho$ into a run $\rho'$ such that any firing of $t_{\initial}$
  happens at the beginning. Let us write $\rho' = t_{\initial}^k
    \sigma$, where $\sigma$ does not contain $t_{\initial}$. We have
  that $\imarked{1} \trans{t_{\initial}^k} \{u \colon k, d \colon
    k\}$. The only transition consuming from $u$ is $t_u$. Since $k$ is
  not a multiple of $c$, and since $t_u$ consumes $c$ tokens from $u$,
  place $u$ can never be emptied. Thus $\{u \colon k, d \colon k\}
    \not \reach \fmarked{1}$.

  Next, let us show that $\imarked{mc} \reach \fmarked{mc}$ for any $m
    \in \Nn$. It follows from
  \begin{multline*}
    \imarked{mc}
    \trans{t_{\initial}^{mc}}
    \{u \colon mc, d \colon mc\} \trans{t_d^{m(c - 1)}}
    \{u \colon mc, d \colon m, \output \colon m(c - 1)\} \\
    \trans{t_u^m} \{\output \colon mc\}.
  \end{multline*}

  Finally, we show that $N_{\neg \text{sound}}$ is not structurally
  sound. It suffices to show that it is $mc$-unsound for all $m \in
    \Nn$. Note that
  \begin{multline*}
    \imarked{mc}
    \trans{t_{\initial}^{mc}}
    \{u \colon mc, d \colon mc\}
    \trans{t_u^m}
    \{d \colon (m - 1)c, \output \colon m\} \\
    \trans{t_d^{((m - 1)c) - 1}}
    \{d \colon 1, \output \colon mc - 1\}.
  \end{multline*}
  No transition is enabled in the latter marking, so $mc$-unsoundness
  follows. \qed
\end{proof}

\end{document}